\documentclass[reqno]{amsart}

\usepackage{a4,amsmath,amssymb,amsthm}
\usepackage{xfrac}

\usepackage{amsaddr}

\usepackage[colorlinks = blue]{hyperref}

\usepackage[T1]{fontenc}
\usepackage[section]{placeins}

\makeatletter
  \def\moverlay{\mathpalette\mov@rlay}
  \def\mov@rlay#1#2{\leavevmode\vtop{%
     \baselineskip\z@skip \lineskiplimit-\maxdimen
     \ialign{\hfil$#1##$\hfil\cr#2\crcr}}}
\makeatother

\usepackage{stmaryrd}

\usepackage{url}
\usepackage{hyperref}

\usepackage{color}

\numberwithin{equation}{section}

\usepackage[mathscr]{eucal}

\usepackage{graphicx}
\graphicspath{pic/}
\usepackage{psfrag}
\usepackage{blindtext}
\usepackage[inline]{enumitem}
\usepackage[labelfont=bf]{subfig}

\newcommand{\Reals}{{\mathbb R}}         
\newcommand{\half}{\frac{1}{2}}         

\newcommand{\HH}{\mathcal H}

\newcommand{\QQ}{\mathcal Q}

\newcommand{\scri}{\mathcal{I}}

\theoremstyle{plain}
\newtheorem{thm}{Theorem}
\newtheorem{cor}[thm]{Corollary}
\newtheorem{lemma}[thm]{Lemma}

\newtheorem{definition}[thm]{Definition}

\newtheorem{remark}[thm]{Remark}
\newtheorem{case}{Case}

\newcounter{mnotecount}[section]

\newcounter{mymnotecount}[section]

\title{Characterization of Null Geodesics on Kerr Spacetimes}
\author[C.~F. Paganini, B. Ruba and M.~A. Oancea]{ Claudio F. Paganini$^{\dagger,*}$,  Blazej Ruba$^\ddagger$ and Marius A. Oancea$^\dagger$} 
\email{claudio.paganini@aei.mpg.de}
\email{blazej.ruba@doctoral.uj.edu.pl} 
\email{marius.oancea@aei.mpg.de}
\address{$^\dagger$Max Planck Institute for Gravitational Physics (Albert Einstein Institute), Am~M\"uhlenberg 1, D-14476 Potsdam, Germany \\
 $^*$Fakult\"at f\"ur Mathematik, Universit\"at Regensburg, D-93040 Regensburg, Germany\\
 $^\ddagger$Institute of Theoretical Physics, Jagiellonian University, Krak\'{o}w, Poland}

\begin{document}

\date{\today \ {\em File:\jobname{.tex}}}

\begin{abstract}
We consider null geodesics in the exterior region of a subextremal Kerr spacetime. We show that most well-known fundamental properties of null geodesics can be represented on one plot. In particular, one can see immediately that the ergoregion and trapping are separated in phase space.
Furthermore, we show that from the point of view of any timelike observer outside of the black hole, trapping can be understood as a smooth set of spacelike directions on the celestial sphere of the observer. Finally, we discuss some applications of these insights.
\end{abstract}

\maketitle

\tableofcontents

\section{Introduction} 

In recent years there has been a lot of progress in the Kerr uniqueness and the Kerr stability problem. Having a thorough understanding of geodesic motion and, in particular, the behavior of null geodesics in Kerr spacetimes is helpful to understand many of the harder problems related to these spacetimes. In this paper, we study the properties of null geodesics in the exterior region of subextremal Kerr spacetimes (i.e.\ with rotation parameter $a \in [0,M)$). The geodesic structure of Kerr spacetimes has been subject of a lot of research. Here, our aim is to give a unified and accessible presentation of the most important properties of geodesics in the Kerr spacetime, with regard to the open problems mentioned above. A Mathematica notebook has been developed for these lecture notes, with intention to help the reader to gain intuition on the influence of various parameters on the geodesic motion, despite the complexity of the underlying equations. It can be downloaded under the permanent link \cite{blazej_marius}. In Section \ref{sec:app} we explain where these plots give useful insights.

An extensive discussion of geodesics in Kerr spacetimes and many further references can be found for example in \cite[p.318]{MR1647491} and \cite{oneill_geometry_2014}. See \cite{teo_spherical_2003} for a nice treatment of the trapped set in Kerr spacetimes, including many explicit plots of trapped null geodesics at different radii. Here, we focus more on global properties of the null geodesics and less on the details of motion. Analysis of the turning points of a dynamical system is a powerful tool to extract information about its global behaviour. For example, in any $1+1$ dimensional system, stable bounded orbits only exist if there exist two distinct turning points in the spatial direction, between which the system can oscillate. For geodesics in Kerr spacetimes, this has been studied in detail by Wilkins \cite{wilkins_bound_1972}. The techniques used here are very close to that paper. Many of the results discussed in these lecture notes are in fact well known. However, our focus is on giving an accessible presentation. We show how all of these well-known properties can be read of from one simple plot. This makes it easier to understand the general behaviour of null geodesics in Kerr spacetimes. We show that various plots provided in the notebook cover the whole space of relevant parameters. A different representation of the forbidden regions in phase space, where (null) geodesics can't exist, can be found in \cite[p.214]{oneill_geometry_2014} and also in \cite{slezakova_geodesic_2006}. The presentation chosen here is adapted to help understand the phase space decomposition used in the proof for the decay of the scalar wave in subextremal Kerr  \cite{2014arXiv1402.7034D}.

In section \ref{sec:sphere}, we introduce the notion of black hole shadows. Null geodesics are often interpreted as light rays of infinite frequency, since they can be obtained by considering the geometrical optics approximation for various field equations, such as the scalar wave equation, Maxwell's equations or linearized gravity \cite{MTW,gravitational_lenses_book} (in some cases, the geometrical optics approximation can be extended to include finite-frequency corrections to null geodesics, see \cite{oancea2019overview} for an overview and \cite{GSHE2020} for recent developments). The shadow of a black hole is thus defined as the innermost trajectory on which light from a background source passing by a black hole can reach the observer.  The first discussion of black hole shadows in Schwarzschild spacetimes can be found in \cite{synge_escape_1966}. For extremal Kerr at infinity it was calculated later in \cite{bardeen_black_1973}. Analyzing the shadows of black holes is of direct physical interest, as the first-ever direct image of a black hole by the Event Horizon Telescope collaboration \cite{collaboration2019first} demonstrated. The opportunity to compare predictions directly with observations has led to a number of advances in the theoretical treatment of black hole shadows in recent years \cite{cunha_shadows_2016,grenzebach_aberrational_2015,grenzebach_photon_2014,grenzebach_photon_2015,hioki_measurement_2009,li_measuring_2014,paganini2018smoothness, mars2017fingerprints,abdujabbarov2015coordinate}, see \cite{cunha2018shadows} for a brief review. 

In particular, we discuss the results of our recent paper \cite{paganini2018smoothness}. We show that for any subextremal Kerr spacetime, including Schwarzschild, the past and the future trapped sets at any point are topologically a circle on the celestial sphere of any observer in the exterior region. This is an important step towards the proof in \cite{mars2017fingerprints}, showing that the radial degeneracy for the shape of the shadow, as it exists in Schwarzschild and for observers on the symmetry axis of Kerr spacetimes, is broken away from the symmetry axis. For future observations, the result in \cite{mars2017fingerprints} implies that, in principle, one  could extract all of the black hole's parameters from a precise measurement of its shadow. However, in these lecture notes we present numerical evidence suggesting that this will require precision which is likely  to remain out of reach for decades to come. 

The significance of Theorem \ref{thm:1} is due to the fact that it describes a property of trapping which does not change when going from Schwarzschild to Kerr. In fact, in a recent paper \cite{cederbaum2019geometry} a much stronger result has been proven, namely that the topology of the area of trapping in phase space remains unchanged when going from Schwarzschild to Kerr. The topology of the area of trapping for some class of higher dimensional black holes has also been analysed in \cite{bugden2019trapped}. Finally, we discuss the results in \cite[p.72]{paganini2018role} that demonstrate that the structures introduced  in the discussion of black hole shadows can be used to prove  that trapping has to exist for general black hole spacetimes.

\subsection*{Overview of this lecture note}

In section \ref{sec:kerr}, we collect some background on the Kerr spacetime. We discuss its symmetries and the  associated  conserved quantities for null geodesic  in section \ref{sec:symandcons}. 
In section \ref{sec:geodeq}, we discuss the geodesic equation in its separated form, focusing on the radial equation and the $\theta$ equation. For the radial equation, we show that many properties of its solutions can be understood from one plot. In section \ref{sec:specialgeod}, we use the analysis of section \ref{sec:geodeq} to discuss the properties of a number of special solutions. In particular, we discuss the radially ingoing and outgoing null geodesics (i.e. the principal null directions) and the trapped set, which are of relevance for the black hole stability problem. Next, we discuss the $\mathbf{T}$-orthogonal null geodesics, which are relevant for the black hole uniqueness problem. In section \ref{sec:sphere}, we prove a theorem concerning the topological structure of the past and future trapped sets. In the same section, we present numerical calculations demonstrating the magnitude of the breaking of the radial degeneracy of the black hole shadows in Kerr spacetimes. In section \ref{sec:cstool}, we discuss results regarding the existence of trapping in general black hole spacetimes. Finally, in section \ref{sec:app}, we discuss how the plots developed for this paper can be used. 

\section{The Kerr Spacetime} \label{sec:kerr} 
The Kerr family of spacetimes describes axially symmetric, stationary and asymptotically flat black hole solutions to the vacuum Einstein field equations.  We use Boyer-Lindquist (BL) coordinates ($t, r,\phi,\theta$), which have the property that the metric components are independent of $\phi$ and $t$. The metric has the form:
\begin{equation}
\label{eq:metric}
g = - \left( 1-\frac{2Mr}{\Sigma} \right) \mathrm{d}t^2 - \frac{2Mar \sin ^2 \theta}{\Sigma} 2 \mathrm{d}t \mathrm{d} \phi + \frac{\Sigma}{\Delta} \mathrm{d} r ^2 + \Sigma \mathrm{d} \theta ^2 + \frac{\sin ^2 \theta}{\Sigma} \mathcal{A} \mathrm{d} \phi ^2,
\end{equation}
where
\begin{equation}
\Sigma = r^2 + a^2 \cos ^2 \theta,
\end{equation}
\begin{equation}
\Delta(r) = r^2 - 2Mr + a^2 = (r-r_+)(r-r_-),
\end{equation}
\begin{equation}
\mathcal A = (r^2+a^2)^2 - a^2 \Delta \sin ^2 \theta.
\end{equation}
The zeros of $\Delta(r)$ are given by: 
\begin{equation}
r_{\pm}=M \pm \sqrt{M^2-a^2} ,
\end{equation}
and correspond to the location of the event horizon at $r=r_+$ and of the Cauchy horizon at $r=r_-$. In the present work we are only interested in the exterior region, hence $r\in(r_+,\infty)$.
 For our considerations it is useful to introduce an orthonormal tetrad. A convenient choice is:
\begin{subequations}\label{eq:tetrad}
\begin{align}
e_0 &= \frac{1}{\sqrt {\Sigma \Delta}} \left( (r^2+a^2)\partial _t + a \partial _ {\phi} \right), \\
e_1&=\sqrt{\frac{\Delta}{\Sigma}} \partial_r, \\
e_2&=\frac{1}{\sqrt{\Sigma}} \partial_{\theta}, \\
e_3&=\frac{1}{\sqrt{\Sigma} \sin \theta} \left( \partial_{\phi} +a \sin ^2 {\theta} \partial_t \right).
\end{align}
\end{subequations}
This frame is a natural choice as the principal null directions can be written in the simple form $e_0 \pm e_1$. These generate congruences of radially outgoing and ingoing null geodesics. We will come back to this fact in section \ref{sec:radialgeod}. In further considerations we will use $e_0$ as the local time direction.
Furthermore, we define the local rotation frequency of the black hole as
\begin{equation}\label{eq:locrotfreq}
\omega (r) = \frac{a}{r^2 +a^2},
\end{equation}
whose limit for $r \searrow r_+$ is called the rotation frequency of the horizon:
\begin{equation}\label{eq:horrotfreq}
\omega _H = \omega (r_+).
\end{equation}
The nameing for $\omega (r)$ is motivated by noting that a particle at rest in the local inertial frame given by the tetrad (\ref{eq:tetrad}) will move in the $\phi$ direction in Boyer-Lindquist coordinates with $\frac{d\phi}{dt}=\omega(r)$ with respect to an observer at rest in  this  frame at infinity.
\subsection{Symmetries and Constants of Motion}\label{sec:symandcons}
The independence of the metric components of the coordinates $t$ and $\phi$ is a manifestation of the presence of two Killing vector fields $(\partial_t)^\nu$, $(\partial_\phi)^\nu$, which both satisfy the Killing equation: 
\begin{equation}\label{eq:killingvf}
\nabla_{(\mu} K_{\nu)}=0
\end{equation} and  commutewith each other. Furthermore, the Kerr spacetime features an irreducible Killing tensor $\sigma_{\mu \nu}$, cf. \cite[p.320]{MR1647491}: a symmetric two-tensor  satisfying the Killing tensor equation:
\begin{equation}
\label{eq:killingt}
\nabla_{(\alpha}\sigma_{\beta \gamma)}=0
\end{equation}
and which can not be written as a linear combination of tensor products of Killing vectors. 
In terms of the tetrad \eqref{eq:tetrad}, the Killing tensor can be written as:
\begin{equation}
\sigma_{\mu \nu}= -a^2 \cos ^2 \theta g_{\mu \nu} + \Sigma \left( (e_2)_{\mu}  (e_2)_{\nu} + (e_3)_{\mu} (e_3)_{\nu} \right).
\label{eq:Killing}
\end{equation}
In the limiting case $a = 0$  it becomes reducible and takes the form: 
\begin{equation}\label{eq:Killing_a0}
\sigma^{(a=0)} = r^4 (\mathrm{d} \theta ^2 + \sin ^2 \mathrm{d} \phi ^2).
\end{equation}
This is the tensor field obtained by taking the sum of the second tensor  product  of the three generators of spherical symmetry. From the geodesic equation:
\begin{equation}\label{eq:geodesic}
\dot \gamma^\mu \nabla_\mu \dot \gamma^\nu =0
\end{equation}
and the Killing equations it follows that a contraction of the tangent vector $\dot \gamma^\mu$ of an affinely parametrized geodesic $\gamma$ with a Killing tensor gives a constant of motion. In Kerr spacetimes those are the mass\footnote{The metric satisfies the Killing tensor equation \eqref{eq:killingt} trivially and therefore it gives us a~conserved quantity as well.}, the energy, the angular momentum with respect to the rotation axis of the black hole and Carter's constant \cite{carter_global_1968}:
\begin{subequations}\label{eq:com}
\begin{align}
-m^2&=g_{\mu\nu}\dot \gamma^\mu\dot \gamma^\nu , \\
E&= - (\partial_t)^\nu\dot \gamma_\nu, \\
L_z&= (\partial_\phi)^\nu \dot \gamma_\nu, \\
K &= \sigma _ {\mu \nu} \dot \gamma ^ {\mu} \dot \gamma ^ {\nu}.
\end{align}
\end{subequations}
It follows from equation (\ref{eq:Killing_a0}) that in a Schwarzschild spacetime $K$ is the square of the total angular momentum of the particle. Carter's constant is non-negative for all time like or null geodesics, which can be seen immediately from equation \eqref{eq:Killing} and the fact that $g_{\mu\nu}\dot \gamma^\mu\dot \gamma^\nu\leq0$ for any future directed causal geodesic. In the case of $a\neq0$ it is even strictly positive for any time like geodesic. It turns out that some combinations of these conserved quantities are more convenient to work with, so we give them their own names: 
\begin{equation}\label{eq:qdef}
Q=K- (aE-L_z)^2,
\end{equation}
\begin{equation}\label{eq:ldef}
L^2 = L_z ^2 + Q.
\end{equation}
One can think of $L^2$ as the total angular momentum of the particle. One can then think of $Q$ as the component of the angular momentum in direction perpendicular to the rotation axis of the black hole.\footnote{The three quantities $K$, $Q$ and $L^2$ are often labeled differently by different authors.} It is important though that these interpretations should not be taken to strictly, because geodesics in Kerr spacetimes do not feature a conserved total angular momentum vector. 
\begin{remark} In contrast to $K$, $Q$ is not positive. However, the equation of motion that we will later define in \eqref{eq:theta}   implies that  for a geodesic to exist at a point  the following equality must be obeyed: 
\begin{equation}
Q \geq -a^2 E^2 \cos ^2 \theta.
\end{equation} 
This follows from the requirement that $\Theta \geq 0$ in \eqref{eq:theta}.

\end{remark} In the case of null-geodesics we can rescale the tangent vector without changing the properties of the geodesic. We will use this to reduce the number of parameters. We define the conserved quotients to be: 
\begin{equation}\label{eq:consquo1}
\mathcal {Q} = \frac{Q}{L_z^2}, \qquad \mathcal{E} = \frac{E}{L_z}.
\end{equation}
An alternative set of conserved quotients is given by: 
\begin{equation}\label{eq:consquo2}
\mathcal{K} = \frac{K}{E^2}, \qquad \mathcal{L} = \frac{L_z}{E}.
\end{equation}
These are more commonly used in the literature as they are more suited for certain calculations. However for the present work we will mostly use the first set. 

\subsubsection*{Properties of the Killing Fields and Tensor}
The vector field $(\partial_\phi)^\nu$ is spacelike for all $r>0$. The vector field $(\partial_t)^\nu$ is timelike  in the asymptotic region, for $r$ sufficiently large. It becomes spacelike in the interior of the \textit{ergoregion}, which is defined by the inequality $g( \partial _t , \partial _t ) \geq 0$, or in terms of BL-coordinates by $-\Delta+a^2 \sin ^2 \theta\leq 0$. The case of equality determines the boundary of the ergoregion, often referred to as the ergosphere,  on which $\partial_t$ is a null vector. Solving for the case of equality we get the radius of the ergosphere to be:
\begin{equation}\label{eq:ergopshere}
r_{ergo}(\theta)= M + \sqrt{M^2-a^2 \cos^2(\theta)}.
\end{equation}
At the equator the ergosphere lies at $r=2M$ while it corresponds to the horizon $r=r_+$ on the rotation axis.\\ 
As mentioned above the two Killing vector fields generate one-parameter groups of isometries. It is natural to ask if the Killing tensor present in the Kerr spacetimes can also be related to some sort of symmetry. This question can be answered using Hamiltonian formalism. 
For a Hamiltonian flow parametrized by $\lambda$ with Hamiltonian $H$ the derivative of any function $f(x,p)$ is given by the Poisson bracket:
\begin{equation}
\frac{\mathrm{d}f}{\mathrm{d}\lambda} = \{ H ,f \} \equiv \frac{\partial H}{\partial p_{\mu}}\frac{\partial f}{\partial x^{\mu}} - \frac{\partial H}{\partial x^{\mu}}\frac{\partial f}{\partial p_{\mu}}.
\end{equation}
Each smooth function on phase space can be taken as a Hamiltonian and therefore gives rise to a local flow. Geodesic motion is generated by the function $-m^2$. $E$ and $L_z$ generate translations in $t$ and $\phi$. In the case $a=0$ the function $K$ generates rotations in the plane orthogonal to the particle's total angular momentum vector with angular velocity equal to the angular momentum of the particle. For rotating black holes it involves a change of all spatial coordinates, but it leaves quantities conserved under geodesic flow invariant. This flow explicitly depends on fiber coordinates $p_{\mu}$ and can not be projected to a symmetry of the spacetime manifold itself. 
 
\section{Geodesic Equations}\label{sec:geodeq}
We now focus our attention on null geodesics. The constants of motion introduced in \eqref{eq:com} can be used to decouple the geodesic equation to a set of four first order ODEs, cf. \cite[p. 242 ]{MR1647491}:
\begin{subequations}\label{eq:eom}
\begin{align}
\Sigma \Delta \dot t &= \mathcal A E - 2Mar L_z,\\
\Sigma \Delta \dot \phi &= 2Mar E + (\Sigma -2Mr) \frac{L_z}{\sin ^2 \theta},\\
\Sigma ^2 \dot r ^2  &= R (r,E,L_z,K)= ((r^2+a^2)E-aL_z)^2-\Delta K,
\label{eq:radial}\\
\Sigma ^2 \dot \theta ^2 &= \Theta (\theta, E, L_z, Q) = Q - \left( \frac{L_z^2}{\sin ^2 \theta} - E^2 a ^2 \right) \cos ^2 \theta,
\label{eq:theta}
\end{align}
\end{subequations}
where the dot denotes differentiation with respect to the affine parameter $\lambda$. \footnote{The radial and the angular equation can be entirely decoupled by introducing a new non-affine parameter $\kappa$ for the geodesics. It is defined by $\frac{\mathrm{d} \kappa}{\mathrm d \lambda}=\frac{1}{\Sigma}$.} For $L_z \neq 0$ the four equations are homogeneous in $L_z$, when written in terms of the conserved quotients. For the radial and the angular equations we have: 
\begin{align}
R(r,E,L_z,Q)&=L_z ^2 R(r,\mathcal E,1,\mathcal Q),\\
\Theta (\theta, E, L_z, Q)& =L_z ^2 \Theta(r,\mathcal E,1,\mathcal Q).
\end{align} 
\begin{remark}
To avoid introducing new functions whenever we change between different sets of conserved quantities we use $R(r,E,L_z,Q)=R(r,E,L_z,K(Q,L_z,E))$.
\end{remark}From the homogeneity of the equations of motion \eqref{eq:eom} we get that the only conserved quantities which affect the dynamics are conserved quotients like $\mathcal E$, $\mathcal Q$ or $\frac{Q}{E^2}$ in the case of $L_z=0$. This is due to the fact that an affine reparametrization $\lambda \mapsto \alpha \lambda$, $\dot \gamma \mapsto \alpha ^{-1} \dot \gamma$ changes the values of $E$, $L_z$ and $Q$ while leaving the trajectories and the aforementioned quotients unchanged. The case $L_z=0$ can be seen as the limit of $\mathcal Q$ and $\mathcal E$ tending to infinity. In these notes we will omit a separate discussion of this case as it is essentially equivalent to the Schwarzschild case and it is not needed for the understanding of the phase space decomposition in 
\cite{2014arXiv1402.7034D}. It is sufficient to consider future directed geodesics as the past directed case follows from the symmetry of the metric when replacing ($t,\phi$) with ($-t,-\phi$).\\
In the Schwarzschild case, a causal geodesic in the external region is future directed if and only if $\dot t > 0$. For Kerr, the suitable condition is $g(\dot \gamma, e_0)\leq 0$. This can be rewritten as
\begin{equation}
\label{eq:future_directed}
E\geq \omega (r) L_z.
\end{equation}
In terms of the conserved quotients this condition takes the form:
\begin{equation}
\mathrm{sgn} ( L_z ) = 
\begin{cases} 
+1 & \mbox{if } \mathcal E > \omega(r) \\
-1 & \mbox{if } \mathcal E < \omega(r) \\
\mathrm{undet.}&\mbox{if } \mathcal E = \omega(r)
\end{cases} 
\label{eq:future_directed_E}
\end{equation}
which eventually allows us to present the pseudo potential, which will be introduced in the next subsection, for the co-rotating and the counter-rotating geodesics in the same plot. 
\subsection{The Radial Equation}\label{sec:radialeq} 
In this section we characterize the radial motion by locating the turning points of a geodesic in $r$ direction. Turning points are characterized by the fact that the component of the tangent vector $\dot \gamma$ in the radial direction satisfies $\dot r =0$. From equation \eqref{eq:radial} we see immediately that the radial turning points are given by the zeros of the radial function $R$. In the following, we will investigate the existence and location of these zeros. In this section we will be working with the conserved quantity $Q$, as it is well suited to describe the phenomena we are interested in here, namely trapping and null geodesics with negative energy $E$.  
\begin{lemma}\label{lem:qneg} $R(r,E,L_z,Q)$ is strictly positive in the exterior region for $Q < 0$.\end{lemma}
\begin{proof}
The radial function can be written as:
\begin{equation}\label{eq:radialpoly}
R(r,E,L_z,Q)=E^2 r^4 + (a^2 E^2 - Q - L_z^2)r^2 +2MKr -a^2 Q,
\end{equation}
which is clearly positive for large $r$. For the proof we make use of the Descartes rule, which states that if the terms of a polynomial with real coefficients are ordered by descending powers, then the number of positive roots is either equal to the number of sign differences between consecutive nonzero coefficients, or is less than it by an even number. Powers with zero coefficient are omitted from the series. For a proof of Descartes rule see for example \cite[p.172]{cohn_algebra._1982}. Applied to \eqref{eq:radialpoly} with $Q<0$ we get that for two zeros of $R$ to exist in $r \in (0, \infty)$ the conserved quantities of the geodesic have to satisfy the inequality:
\begin{equation}
a^2 E^2 - Q - L_z^2 < 0.
\label{eq:ineq}
\end{equation}
 Otherwise there are no zeros at all and the proposition is true. Assume the contrary. Then for geodesics with certain parameters to exist at a given point, additionally to $R\geq 0$ we also need to have that $ \Theta \geq 0$. Applying this condition to equation (\ref{eq:theta}) and combining it with inequality \eqref{eq:ineq} we obtain the following estimate:
\begin{equation}
- \cos ^4 \theta a^2 E^2 \geq -Q > 0.
\end{equation}
This is clearly a contradiction.

\end{proof}
\noindent  Lemma \ref{lem:qneg} tells us that geodesics with $Q<0$ either come from past null infinity $\scri^-$ and cross the future event horizon $\HH^+$ or come out of the past event horizon $\HH^-$ and go to future null infinity $\scri^+$. We will discuss the property of these geodesics in section \ref{sec:radialgeod}. For the rest of this section we will restrict to the case of $Q\geq0$. Note that even though some of the discussions and proofs might be simpler when working with $\mathcal K$, we are going to work with $\mathcal Q$ for all discussions that feed into the plot in section \ref{sec:app}. This seems to be the natural choice to describe trapping and the ergoregion in phase space.\\
To find the essential properties of the radial motion, we use pseudo potentials. The pseudo potential $V(r,\mathcal Q)$ is defined as the value of $\mathcal E$ such that the radius $r$ is a~turning point. In other words it is a solution to the equation:
\begin{equation}\label{eq:pseudopotdef}
R(r,V(r,\mathcal Q),1, \mathcal Q) = \Sigma ^2 \dot r ^2 = 0.
\end{equation}
This equation is quadratic in $V(r,\QQ)$ and for non-negative $\QQ$ there exist two real solutions at every radius, denoted by $V_{\pm}$. They are given by:
\begin{equation}
V_{\pm} (r,\mathcal Q) = \frac{2Mar \pm \sqrt{r \Delta ((1+\mathcal Q) r^3 + a^2 \mathcal Q (r+2M))}}{r [ r(r^2+a^2)+2Ma^2 ]}.
\label{eq:pseudopot}
\end{equation}
\begin{remark}
The pseudo potentials should not be mistaken for potentials known from classical mechanics, where the equation of motion is given by $\half \dot x^2 + V(x)=E$. However the potentials of classical mechanics can always be considered as pseudo potentials in the above sense.
\end{remark} 
\noindent The radial function can be rewritten as: 
\begin{equation}\label{eq:decomppot}
R(r,\mathcal E,1, \mathcal Q) = L_z^2 r[r(r^2+a^2) +2Ma^2] \left( \mathcal E-V_+(r,\mathcal Q) \right)  \left( \mathcal E - V_- (r,\mathcal Q) \right).
\end{equation}
This form of $R$ reveals the significance of the pseudo potentials: The only turning points that can exist for fixed $\QQ>0$ are those where either $\mathcal E = V _{+} (r,\QQ)$ or  $\mathcal E = V _{-} (r,\QQ)$.  Analyzing the properties of the $V_\pm$ allows us to extract all the information we are interested in.\\
First we note that for $r$ big enough we have that $V_+>0$ and $V_-<0$ for all $\QQ\geq0$. However in the limit we have that: 
\begin{equation}\label{eq:liminfty}
\lim _ {r \to \infty} V_ {\pm} = 0.
\end{equation}
At the horizon the limit of the pseudo potential and its derivative are given by:
\begin{equation}\label{eq:limhorizon}
\lim _{r \to r_+} V _ {\pm} (r) = \omega _H
\end{equation}
\begin{equation}
\lim _ {r \to r_+} \frac{dV_{\pm}}{dr}(r) = \pm \infty.
\end{equation}
\begin{lemma}\label{lem:onemax}
For a fixed value of $\QQ$ the pseudo potentials $V_\pm (r,\QQ)$ have exactly one extremum as a function of $r$ in the interval $(r_+,\infty)$. 
\end{lemma}
\begin{proof}
It is clear from the above properties that  $V_+$ ($V_-$) has at least one maximum (minimum) in the DOC. From the fact that the two pseudo potentials have the same limiting value at $\infty$ and at $r_+$ together with \eqref{eq:decomppot} we get that in both limits we have that $R(r,\mathcal E,1,\mathcal Q)\geq0$. Therefore $R(r,\mathcal E,1,\mathcal Q)$ has to have an even number of zeros in the interval $(r_+, \infty)$. Given the fact that $R(r,\mathcal E,1,\mathcal Q)$ is a fourth order polynomial it can have at most 4 zeros. From the asymptotic behaviour of the potentials $V_{\pm}$ we get that they need to have an odd number of extrema. Therefore if for some value of $\mathcal Q$ one of the potentials has more than one extremum there exists $\mathcal E$ such that $R(r,\mathcal E,1,\mathcal Q)$ has three zeros in $r \in (r_+,\infty)$. Applying Descartes rule to \eqref{eq:radial} we infer that $R(r,\mathcal E,1,\mathcal Q)$ can have at most three zeros in $r \in [0,\infty)$. But $R(0,\mathcal E,1,\mathcal Q) \leq 0$ and $R(r_+,\mathcal E,1,\mathcal Q) \geq 0$. Hence there is at least one zero of $R(r,\mathcal E,1,\mathcal Q)$ in the interval $[0,r_+]$, so it is impossible for $R(r,\mathcal E,1,\mathcal Q)$ to have three zeros in $(r_+,\infty)$. From that it follows that  $V_\pm$ can both only have one extremum in that interval. 
\end{proof}
\noindent Stationary points occur at the extrema of the pseudo potentials. So Lemma \ref{lem:onemax} tells us that for every fixed value of $\mathcal Q \geq 0$ there exist exactly two spherical geodesics with radii $r=r_{trap} ^{\pm}$ and energies $\mathcal{E}_{trap \pm} =V _{\pm} (r_{trap} ^ {\pm})$. They will be studied in depth in section \ref{sec:trapped}. Bounded geodesics with non-constant $r$ would only be possible between two extrema of one of the pseudo potentials. These are excluded by the Lemma. \\
From (\ref{eq:radial}) we have that for any geodesics to exist we have to have $R(r,\mathcal E,1,\mathcal Q) \geq 0$. This condition is satisfied except if  $V_- (r,\QQ) < \mathcal E < V_+(r,\QQ)$. This set is therefore a forbidden region in the $(r, \mathcal E)$ plane.  Furthermore it follows that $R(r,\mathcal E,1,\mathcal Q)\leq0$ for $\mathcal E=\omega (r)$ with equality only in the limits $r \rightarrow r_+$ and $r \rightarrow \infty$. Therefore we have that: 
\begin{equation}
V_-(r) \leq \omega (r) \leq V_+ (r).
\label{eq:omega_forbidden}
\end{equation}
Again, the equality can occur only at the horizon and in the limit $r \to \infty$. This fact combined with \eqref{eq:future_directed_E} shows that for future pointing null geodesics
\begin{equation}
\mathrm{sgn} ( L_z ) = 
\begin{cases} 
+1 & \mbox{if } \mathcal E \geq V_+(r), \\
-1 & \mbox{if } \mathcal E \leq V_-(r).
\end{cases} 
\label{eq:Lz_sign}
\end{equation}
Therefore the pseudopotential $V_+$ determines the behaviour of co-rotating null geodesics and $V_-$ that of counter-rotating ones. Furthermore it is worth noting that $\mathcal E \geq V_+(r)$ implies $E>0$. 
Finally we observe that for every fixed radius $r\geq r_+$ we get from inspection of \eqref{eq:pseudopot} that: 
\begin{equation}
 \frac {\partial V_ {-}}{\partial \mathcal Q} \leq 0 \leq \frac {\partial V_ {+}}{\partial \mathcal Q}
 \label{eq:Forbidden_Region_Expansion} 
\end{equation}
holds. The equality in the relation occurs again only in the limits $r \to r_+$ and  $r \to \infty$. This means that for every radius $r>r_+$ the range of forbidden values of $\mathcal{E}$ is strictly expanding as $\mathcal Q$ increases. This fact will be quite useful for the considerations in section \ref{sec:trapped}.

\subsection{The \texorpdfstring{$\theta$}{theta} Equation}\label{sec:thetaeq}

In Schwarzschild spacetimes, due to spherical symmetry the motion of any geodesics is contained in a plane. This means that for every geodesic there exists a spherical coordinate system in which it is constrained to the equatorial plane $\theta = \frac{\pi}{2}$. This is no longer true in Kerr spacetimes, but most geodesics are still constrained in $\theta$ direction. The allowed range of $\theta$ is obtained by solving the inequality $\Theta (\theta,E,L_z,Q) \geq 0$. After multiplication with $\sin^2 (\theta)$, $\Theta (\theta,E,L_z,Q)$ can be expressed as a quadratic polynomial in the variable $\cos^2(\theta)$. Hence $\Theta (\theta,E,L_z,Q) = 0$ has two solutions given by:
\begin{equation}
\cos ^2 \theta _{turn} = \frac{a^2 E^2 - L_z^2-Q \pm \sqrt{(a^2E^2-L_z^2-Q)^2+4a^2E^2Q}}{2a^2E^2}.
\end{equation}
For $Q > 0$ only the solution with the plus sign is relevant and the motion will always be contained in a band $\theta _{min} < \theta_{eq} < \theta _{max}$ symmetric about the equator $\theta_{eq} = \frac{\pi}{2}$. As $|L_z|$ increases, this band shrinks. In fact only in the case $L_z=0$ it is possible for a geodesic to reach the poles $\theta=0$, $\theta=\pi$. Otherwise $\Theta(\theta,E,L_z,Q)$ blows up to $- \infty$ there. If $Q < 0$ both solutions are positive and the inclination of the geodesic with respect to the equator is also constrained away from the equator, so either $ \theta_{eq} <  \theta _{min}< \theta _{max}$ or $ \theta_{eq} >  \theta _{max} > \theta _{min}$ . These trajectories are contained in disjoint bands which intersect neither the equator nor the pole. These bands can degenerate to circles, i.e. there exist null geodesics which stay at $\theta=\mathrm{const.}$ The relevance of these trajectories and their connection to the Schwarzschild case will be discussed in the next section. All possibilities for the potentials that constrain the motion in $\theta$ direction are summarized in the Figure \ref{fig:ThetaFunction}.

\begin{figure}[t!]
\centering
\includegraphics[width=120mm]{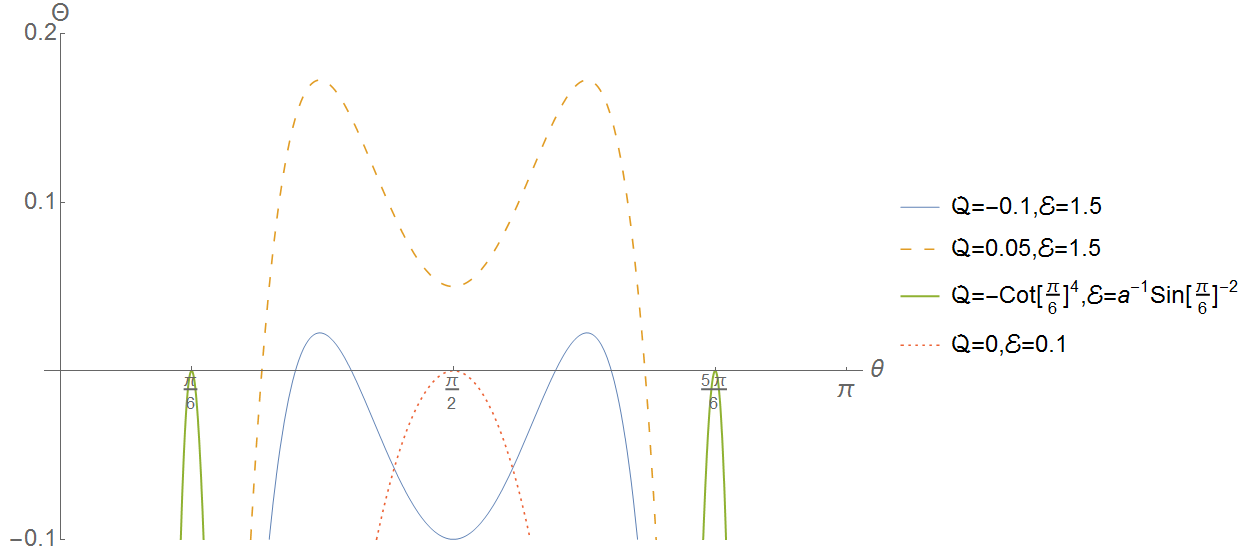}
\caption{This figure shows shapes of function $\frac{\Theta}{L_z^2}$ for four choices of values of conserved quotients.}
\label{fig:ThetaFunction}
\end{figure}

\section{Special Geodesics}\label{sec:specialgeod}
We will now apply the discussions of the last section to describe a number of special geodesics in Kerr geometries. All of these are in some way related to either the black hole stability problem or the black hole uniqueness problem. 

\subsection{Radially In-/Out-going Null Geodesics}\label{sec:radialgeod}
In this section we find geodesics which generalize the radially ingoing and outgoing congruences in Schwarzschild spacetimes. In section \ref{sec:radialeq} we saw that the geodesics with $Q<0$ extend from the horizon to infinity. In section \ref{sec:thetaeq} we saw that $Q<0$ is again a special case, as these null geodesics can never intersect the equator and in the extreme case are even constrained to a fixed value of $\theta$. At first this behaviour seems odd, but a similar situation can be observed in the Schwarzschild case. If we look at geodesics which move in a plane with inclination $\theta_0$  with respect to the equatorial plane, we see that there exists a set of null geodesics with similar properties as the ones with $Q<0$ in Kerr. It is clear that the radially ingoing geodesics which move orthogonally to the axis around which the plane of motion was rotated, move at fixed $\theta$ value, namely that at which the plane is inclined with respect to the equatorial plane, hence $\theta= \pi/2 \pm \theta_0$. Furthermore some null geodesics reach the horizon before intersecting the equatorial plane. They don't necessarily move at fixed $\theta$ but their motion in $\theta$ direction is still constrained away form the equatorial plane and away from the poles of the coordinate system.\\
Now we want to investigate the null geodesics which move at a fixed $\theta$ in Kerr. Demanding that $\theta=\mathrm{const.}$ is equivalent to requiring $\Theta=\frac{d}{d\theta}\Theta =0$. From these conditions we obtain:
\begin{subequations}
\begin{align}
L_z&= \pm  aE \sin ^2 \theta, \\
Q&=-a^2 E^2 \cos ^4 \theta, \\
K&=0.
\end{align}
\label{eq:geodesic_principal}
\end{subequations}
By choosing the plus sign for $L_z$ in the above equation, it follows  from the remaining equations of motion that:
\begin{subequations}
\begin{align}
\frac{\dot \phi}{\dot t}=\frac{d\phi}{dt} &= \omega(r), \\
\frac{\dot r}{\dot t}=\frac{dr}{dt} &= \pm \frac{\Delta}{r^2+a^2}.
\end{align}
\end{subequations}
This congruence is generated by the principal null directions $e_0 \pm e_1$. When choosing the minus sign for $L_z$ in \eqref{eq:geodesic_principal} the remaining equations of motions can not be simplified in a similar manner. In the case $a=0$ these are the radially in-/outgoing null geodesics. An interesting observation is that along these geodesics the inner product of the $(\partial_t)^\mu$ vector field is monotone. A simple calculation shows that:
\begin{equation}
\dot\gamma^\mu\nabla_\mu ((\partial_t)^\nu(\partial_t)_\nu) = \dot r \frac{2M (r^2 -a^2 \cos^2 \theta)}{\Sigma^2}+ \dot\theta \frac{2Ma^2 r \sin 2\theta}{\Sigma^2}.
\end{equation}
For the principal null congruence we have $\dot\theta=0$, the coefficient of $\dot r$ is positive and there is no turning point in $r$. This property might be interesting in the context of the black hole uniqueness problem. If one could show a similar monotonicity statement for a congruence of null geodesics in general stationary black hole spacetimes, one could conclude that the ergosphere in such spacetimes has only one connected component enclosing the horizon. This is a necessary condition if one wants to show that no trapped $\mathbf{T}$-orthogonal null geodesics can exist in that case. 

\subsection{The Trapped Set}\label{sec:trapped}
One of the most important features of geodesic motion in black hole spacetimes is the possibility of trapping. A geodesic is called trapped if its motion is bounded in a spatially compact region away from the horizon. In~Kerr in~Boyer-Lindquist coordinates this corresponds to the geodesics motion being bounded in $r$ direction. This is only possible if $r=\mathrm{const.}$ or if the motion is between two turning points of the radial motion. For null geodesics in Kerr we ruled out the second option in Lemma \ref{lem:onemax}. We will now discuss orbits of constant radius.\footnote{Null geodesics of constant radius are often referred to as "spherical null geodesics" but it is important to note that $r=\mathrm{const.}$ does not imply that the whole sphere is accessible for such geodesics.} These null geodesics are stationary points of the radial motion, hence null geodesics with $\dot r= \ddot r=0$. Dividing equation \eqref{eq:radial} by $\Sigma$ and taking the derivative with respect to $\lambda$ we see that this condition is equivalent to $R(r)=\frac{d}{dr}R(r)=0$. The solutions to these equations can be parametrized explicitly by, cf. \cite{teo_spherical_2003}:
\begin{equation}
\mathcal E _{trap} (r) = - \frac{a(r-M)}{A(r)} = \omega (r) \left( 1 - \frac{2r \Delta }{A(r)} \right)
\label{eq:E_trap}
\end{equation}
\begin{equation}
\mathcal Q _{trap} (r) = - \frac{B(r)}{A^2(r)}
\label{eq:Q_trap}
\end{equation}
\begin{equation}
A(r)=r^3-3Mr^2+a^2r+a^2M=(r-r_3)P_2(r)
\end{equation}
\begin{equation}
B(r)=r^3(r^3-6Mr^2+9M^2r-4a^2M)=(r-r_1)(r-r_2)P_4(r)
\end{equation}
where $P_2$ and $P_4$ are polynomials in $r$, quadratic and quartic respectively, which are strictly positive in the DOC. The following three radii are particularly important:
\begin{equation}
r_1=2M\left(1+ \cos \left( \frac{2}{3} \arccos\left(-\frac{a}{M}\right) \right) \right)
\end{equation}
\begin{equation}
r_2=2M\left(1+ \cos \left( \frac{2}{3} \arccos\left(\frac{a}{M}\right) \right) \right)
\end{equation}
\begin{equation}
r_3=M+2 \sqrt{M^2-\frac{a^2}{3}} \cos \left ( \frac{1}{3} \arccos
\left( \frac{M(M^2-a^2)}{(M^2-\frac{a^2}{3})^{\frac{3}{2}}} \right) \right)
\end{equation}
satisfying the inequalities:
\begin{equation}
M<r_+<r_1<r_3<r_2<4M 
\label{eq:radial_inequalities}
\end{equation}
for $a\in(0,M)$. Orbits of constant radius are allowed only inside the interval $[r_1,r_2]$, because outside of it $Q$ would have to be negative. This possibility has already been excluded in Lemma \ref{lem:qneg}. The boundaries of the interval at $r=r_1$ and $r=r_2$ correspond to circular geodesics constrained to the equatorial plane with $Q=0$. The trapped null geodesics at $r=r_3$ have $L_z=0$ which is the reason why the  functions $\mathcal E _{trap}$ and $\mathcal Q _{trap}$ blow up there. From the second representation in (\ref{eq:E_trap}) we see that $\mathcal{E} _{trap} (r) - \omega (r)$ is positive in $[r_1,r_3)$ and negative in $(r_3,r_2]$. Combined with (\ref{eq:Lz_sign}) this implies that the stationary points in $[r_1,r_3)$ correspond to extrema of $V_+$ and the stationary points in $(r_3,r_2]$ correspond to extrema of $V_-$. In Lemma \ref{lem:onemax} we showed that $V_+$ and $V_-$ both have exactly one extremum. Since extrema of the pseudo potentials always correspond to orbits of constant radius, we get that the extrema of $V_+(r,\QQ)$ and $V_-(r,\QQ)$ have to be within the intervals $[r_1,r_3)$ and $(r_3,r_2]$ respectively for any value of $\QQ$. In Figure \ref{fig:radii} we plot the behaviour of these intervals as a function of $a/M$.\\
\begin{figure}[t!]
\centering
\includegraphics[width=120mm]{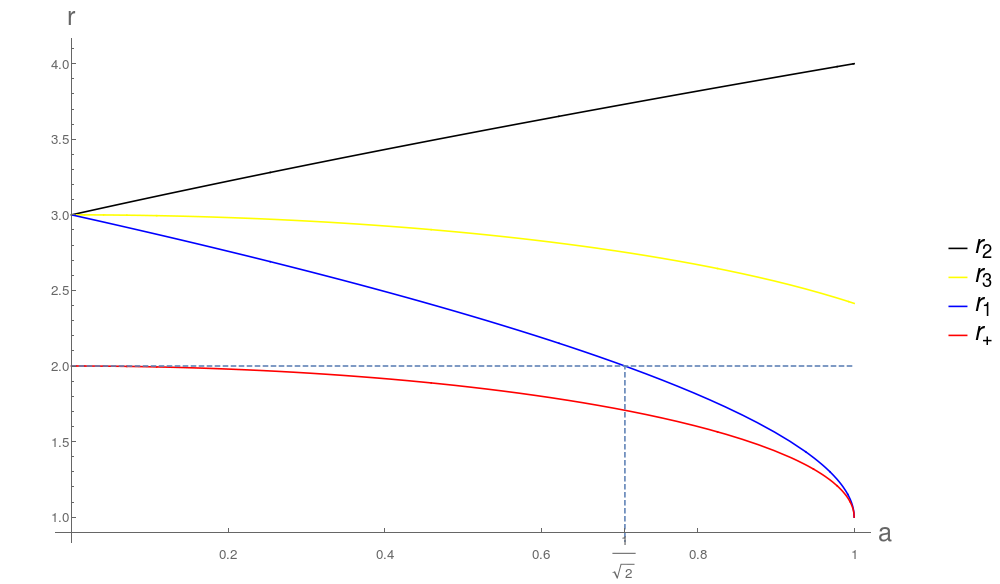}
\caption{Plot of the relation between the radius of the equatorial trapped null geodesics at $r_1$ and $r_2$, the trapped null geodesic with $L_z=0$ at $r_3$ and the horizon at $r_+$ for all values of $a$.}
\label{fig:radii}
\end{figure}

\noindent We now know that the maps given by: 
\[ [0,\infty ) \ni \mathcal Q \mapsto r_{trap} ^ + \in [r_1,r_3)  \]
\[ [0,\infty ) \ni \mathcal Q \mapsto r_{trap} ^ - \in (r_3,r_2] \]
which take $\mathcal Q$ into radii of trapped geodesics corresponding to the unique maximum of $V_+(r,\mathcal Q)$ and minimum of $V_-(r, \mathcal Q)$ respectively are one-to-one and therefore monotone. By using \eqref{eq:Q_trap} the sign of their derivatives can be easily evaluated in some $\epsilon$-neighbourhood of $r=r_3$ where the term of highest order in $\frac{1}{r-r_3}$ dominates:
\begin{equation}
 \frac{\partial r_{trap} ^{-}}{\partial \mathcal Q} < 0 < \frac{\partial r_{trap} ^{+}}{\partial \mathcal Q}.
\end{equation}
By the equation (\ref{eq:Forbidden_Region_Expansion}) and the fact that radii of trapping always correspond to global extrema of the pseudo potentials we get that: 
\begin{equation}
\frac{\partial}{\partial \mathcal Q} \mathcal E _{trap} (r _{trap} ^{-} (\mathcal Q))<0<\frac{\partial}{\partial \mathcal Q} \mathcal E _{trap} (r _{trap} ^{+} (\mathcal Q)).
\end{equation}
Using the chain rule and combining these two facts we obtain:
\begin{equation}
\frac{\partial \mathcal E _{trap}}{\partial r} > 0.
\end{equation}
These inequalities provide an important piece of the picture of the influence of $\mathcal Q$ on the trapped geodesics. We have $\mathcal Q=0$ for the outermost circular geodesics and as we increase it, the radii of trapping converge towards $r=r_3$ while $\mathcal E$ blows up to $\pm \infty$, with the sign depending on the direction from which we approach $r_3$. We can also desciribe the function $\mathcal E _{trap} (r)$: it starts with some finite positive value at $r=r_1$ and increases monotonically to $+ \infty$ as $r$ approaches $r_3$. There it jumps to $- \infty$ and increases again to a finite negative value at $r=r_2$. \\
It is interesting to ask what region in physical space is accessible for trapped geodesics. By plugging (\ref{eq:Q_trap}) and (\ref{eq:E_trap}) into the equation $\Theta = 0$ we get that for a geodesic with $r=\mathrm{const.}$:

\begin{equation}
\cos ^2 \theta_{turn}=\frac{2\sqrt{Mr^2 \Delta (2r^3 -3Mr^2 +a^2 M)}-r(r^3-3M^2 r +2 a^2 M)}{a^2(r-M)^2}
\label{eq:theta_turn}
\end{equation}
holds. This gives two turning points in $\theta$ direction which are symmetric about the equatorial plane.
The whole region of trapping in the $(r,\theta)$ plane is bounded by curves defined implicitly by (\ref{eq:theta_turn}) and $r_1 \leq r \leq r_2$. Figure (\ref{fig:AreaOfTrapping}) presents this set for a particular value of $a$.

\begin{figure}[t!]
\centering
\includegraphics[width=120mm]{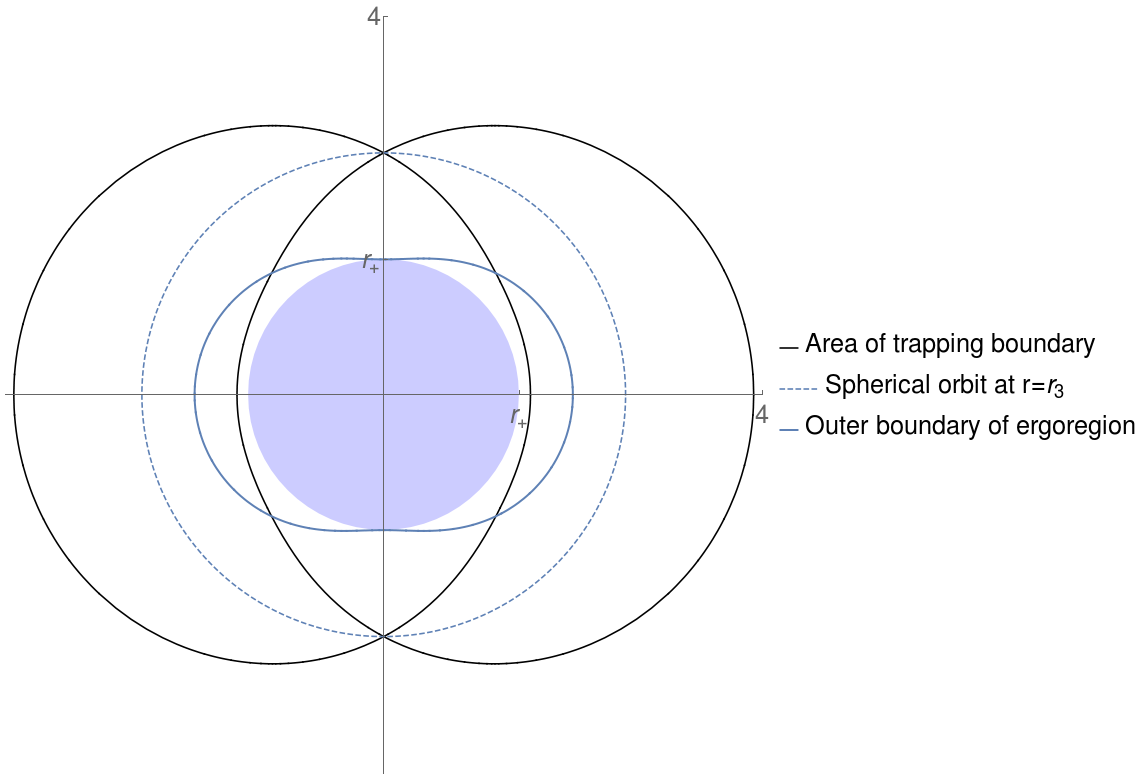}
\caption{The region accessible for trapped null geodesics for $a=0.902$. The shaded region represents the black hole, $r \leq r_+$. The only qualitative change in this picture occurs at $a=\frac{1}{\sqrt{2}}$, because at this value the region of trapping starts intersecting with the ergoregion.}
\label{fig:AreaOfTrapping}
\end{figure}
\begin{remark}
Two warnings: \begin{enumerate}
    \item  One has to be careful when interpreting Figure \ref{fig:AreaOfTrapping} (and the plots in the Mathematica notebook). Despite the fact that the region in physical space occupies finite range of $r$ values, every individual trapped null geodesic is still constrained to a fixed radius. For an insight on what those trajectories look like in detail we recommend the study of \cite{teo_spherical_2003}.
    \item When taking $a\rightarrow M$ in the Mathematica notebook the ergosphere appears to develop a kink on the rotation axis. This is an artifact of the coordinate system, as the ergosphere coincides with the horizon there and is thus orthogonal to itself. 
\end{enumerate}
\end{remark}
\subsection{\texorpdfstring{$\mathbf T$}{T}-Orthogonal Null Geodesics}\label{sec:tortho}
In the ergoregion there exist null geodesics with negative values of $E$. In physical space they are constrained to the region defined by equation \eqref{eq:ergopshere}. From Lemma \ref{lem:qneg} we know that geodesics with $Q<0$ reach either $\scri^+$ or come from $\scri^-$, and can therefore not have negative values of $E$. This allows us to use the pseudo potentials to give  a more precise characterization of the ergoregion in phase space. It is located in the region where $V_-(\mathcal Q) >0$, between $\mathcal{E}=0$ and $V_-(\mathcal{Q})$. An immediate consequence of that is, that all future pointing null geodesics with negative $\mathcal{E}$ begin at the past event horizon and end at the future event horizon. Furthermore they must have $L_z<0$. Those null geodesics with $\mathcal{E}=0$ can reach the boundary of the ergoregion. In this case equation \eqref{eq:theta} gives us, that:
\begin{equation}
\mathcal{Q} = \frac{\cos^2(\theta_{max})}{\sin^2(\theta_{max})}.
\end{equation}
When calculating the turning points from equation \eqref{eq:radial} we get that: 
\begin{equation}
\sin^2(\theta_{max}) R\left(r,0,1,\frac{\cos^2(\theta_{max})}{\sin^2(\theta_{max})}\right)=-r^2 + 2Mr -a^2\cos^2(\theta_{max})=0.
\end{equation}
The only solution to this equation in the exterior region is: 
\begin{equation}
r_{turn}(\theta_{max})= M+\sqrt{M^2-a^2\cos^2(\theta_{max}) }
\end{equation}
which is exactly the location of the ergo sphere \eqref{eq:ergopshere}. So $V_-(\mathcal{Q})>0$ can be considered as the boundary of the ergoregion in phase space. From this considerations we see immediately that $\mathbf{T}$-orthogonal null geodesics are clearly non-trapped in Kerr. In fact there do not even exist any trapped null geodesics orthogonal to:
\begin{equation}
K^\nu =(\partial_t)^\nu+ \epsilon_{min} (\partial_\phi)^\nu
\end{equation} where $\epsilon_{min}= min [ |V_+(0, r_1)|,|V_-(0, r_2)|]$. See \cite{claudio} for informal lecture notes on the $\mathbf{T}$-orthogonal trapping problem in stationary black hole spacetimes. 

\section{Trapping as a Set of Directions}\label{sec:sphere}
 In this section we will link the previous discussion to the black hole shadows. We introduce a more formal framework for the discussion. This allows us to give a~more technical discussion of the trapped sets in Schwarzschild and Kerr black holes as well as proving the existence of trapping for general black hole spacetimes. 
\subsection{Framework}
First we have to introduce the basic framework and notations. Let $\mathcal{M}$ be a smooth manifold with Lorenzian metric $g$. At any point $p$ in $\mathcal{M}$ one can find an orthonormal basis $(e_0,e_1,e_2,e_3)$ for the tangent space, with $e_0$ being the timelike direction. It is sufficient to treat only future directed null geodesics as the past directed ones are identical, up to a sign flip in the parametrization. The~tangent vector to any future pointing null geodesic can be written as:
\begin{equation}\label{eq:tangentplus}
\dot\gamma( k|_p)|_p =\alpha \cdot ( e_0+ k_1 e_1+ k_2 e_2+ k_3 e_3)
\end{equation} 
where $\alpha= -g(\dot\gamma,e_0)$ and $k =(k_1,k_2,k_3)$ satisfies $|k|^2 =1$, hence $k\in S^2$. The~geodesic is independent of the scaling of the tangent vector as this corresponds to an affine reparametrization for the null geodesic. We will therefore set $\alpha=1$ in the following discussion. The $S^2$ is often referred to as the celestial sphere of a timelike observer at $p$ whose tangent vector is given by $e_0$, cf. \cite[p.8]{penrose_spinors_1987}.\\
For the further discussions we fix the tetrads. We can make the following definition:
\begin{definition}\label{def:gammak}Let $\gamma (k|_p)$ denote the null geodesic through $p$ for which the tangent vector at $p$ is given by equation \eqref{eq:tangentplus}.
\end{definition}
\noindent It is clear that $\gamma(k_a|_p)$ and $\gamma (k_b|_p)$ are equivalent up to parametrization if $k_a=k_b$. Suppose now that $\mathcal{M}$ is the exterior region of a black hole spacetime with complete future and past null infinity $\scri^\pm$ and a boundary given by the future and past event horizon $\HH^+\cup\HH^-$. We can then define the following sets on $S^2$ at every point $p$.
\begin{definition}\label{def:futinf}
The future infalling set: $\Omega_{\mathcal{H}^+}(p):=  \{k\in S^2 | \gamma(k|_p)\cap\mathcal{H}^+ \neq \emptyset \} $.\\
The future escaping set: $\Omega_{\mathcal{I}^+}(p):=  \{k\in S^2| \gamma(k|_p)\cap\mathcal{I}^+ \neq \emptyset \}$ .\\
The future trapped set: $\mathbb{T}_+(p) := \{k\in S^2 | \gamma(k|_p)\cap(\mathcal{H}^+\cup \mathcal{I}^+)  = \emptyset \}$.\\
The past infalling set: $\Omega_{\mathcal{H}^-}(p):=  \{k\in S^2 | \gamma(k|_p)\cap\mathcal{H}^- \neq \emptyset \}$.\\
The past escaping set: $\Omega_{\mathcal{I}^-}(p):= \{k\in S^2 | \gamma(k|_p)\cap\mathcal{I}^- \neq \emptyset \}$.\\
The past trapped set: $\mathbb{T}_- (p):= \{k\in S^2 | \gamma(k|_p)\cap(\mathcal{H}^-\cup \mathcal{I}^-)  = \emptyset \}$
\end{definition}
\noindent An illustration of these sets is given in Figure \ref{fig:causal}. We finish the section by defining the trapped set to be: 
\begin{definition}
The trapped set: $\mathbb{T}(p):= \mathbb{T}_+(p)\cap\mathbb{T}_-(p)$.
\end{definition}
\noindent The region of trapping in the manifold $\mathcal{M}$ is then given by:
\begin{definition}
Region of trapping: $\mathcal{A}:= \{p\in \mathcal{M}| \mathbb{T}(p)\neq \emptyset\}$.
\end{definition}

\begin{figure}[t!]
\centering
 \subfloat[\label{fig:causal1}]{%
 \includegraphics[width=.47\textwidth]{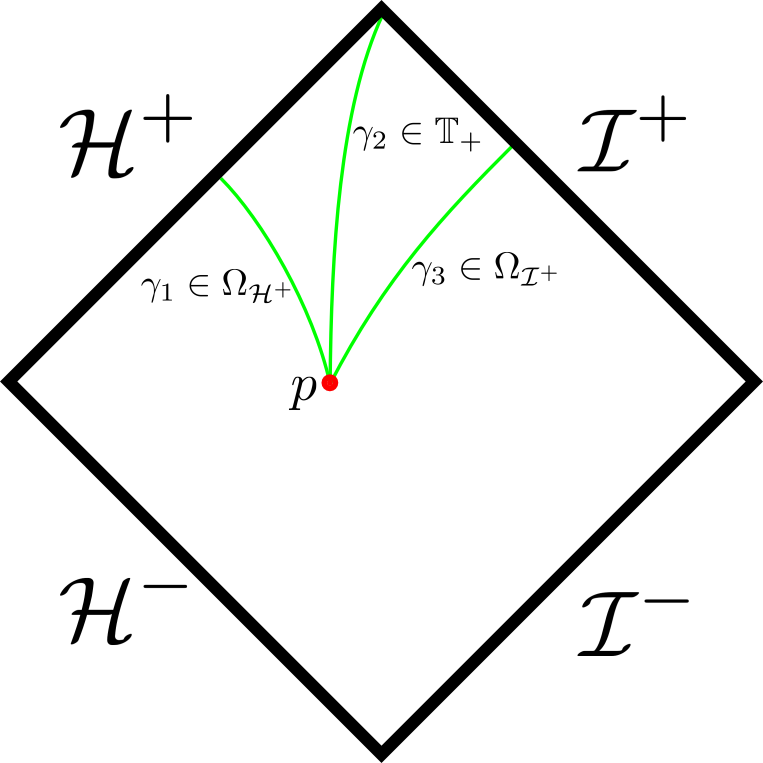}}\hfill
\subfloat[\label{fig:causal2}]{%
  \includegraphics[width=.47\textwidth]{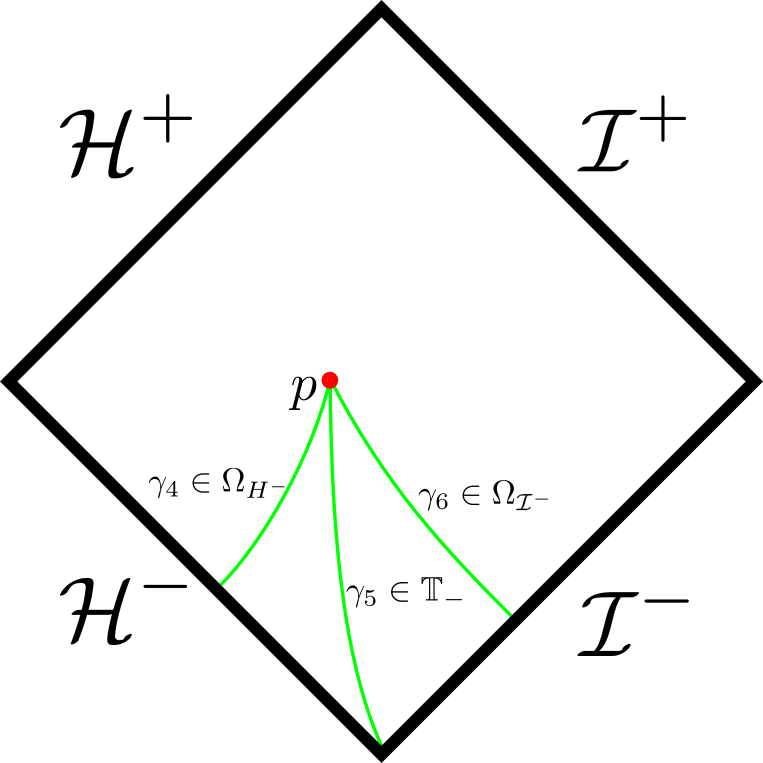}}\hfill
\caption{Conformal diagrams giving a schematic representation of elements of the sets in Definition \ref{def:futinf}. }
\label{fig:causal}
\end{figure}

\subsection{The trapped sets in Schwarzschild}\label{sec:ssshadow}
The discussion of Schwarzschild serves as an easy example for the various concepts.
\begin{definition}
We refer to the set $\Omega_{\mathcal{H}^-}(p)\cup \mathbb{T}_-(p)$ as the shadow of the black hole. 
\end{definition}
\noindent However for any practical purposes information about its boundary which is given by $\mathbb{T}_-(p)$ is good enough. An observer in the exterior region can only observe light from background sources along directions in the set $\Omega_{\mathcal{I}^-}(p)$. From the radial equation we get immediately that if $k=(k_1,k_2,k_3)\in \mathbb{T}_+(p)$ then  $k=(-k_1,k_2,k_3)\in \mathbb{T}_-(p)$. Hence the properties of the past and the future sets are equivalent. This is true both in Schwarzschild and Kerr. \\An explicit formula for the shadow of a Schwarzschild black hole was first given in \cite{synge_escape_1966}. In Schwarzschild the orthonormal tetrad \eqref{eq:tetrad} reduces to:

\begin{align}\label{eq:sstetrad}
e_0 &= \frac{1}{\sqrt {1-2M/r}} \partial _t ,&\qquad e_1&=\sqrt{1-2M/r}\partial_r, \\ \nonumber
e_2&=\frac{1}{r} \partial_{\theta},&\qquad e_3&=\frac{1}{r \sin \theta} \partial_{\phi}. 
\end{align}

To determine the structure of $\mathbb{T}_\pm(p)$ in the Schwarzschild case, it is sufficient to consider $p$ in the equatorial plane and $k=(\cos\Psi,0,\sin\Psi)$  with $\Psi \in [0,\pi]$. The entire sets $\mathbb{T}_\pm(p)$ are then obtained by rotating around the $e_1$ direction. Note that from the tetrad it is obvious that $E(k)=E(r)$ is independent of $\Psi$. On the other hand $L_z(k)$ is zero for $\Psi=\{0,\pi\}$ and maximal for $\Psi=\pi/2$. Away from that maximum, $L_z$ is a monotone function of $\Psi$. Note that the geodesic that corresponds to $\Psi=\pi/2$ has $k_1=0$ and thus a radial turning point. Thus the $E/L_z$ value of this geodesic corresponds to the minimum value any geodesic can have at this point in the manifold. For $r\neq 3M$ this is smaller then the value of trapping and thus there exist two $k$ with the property that $E/L_z(k)= 1/\sqrt {27 M^2}$. One of them has $k_1>0$ and therefore $\dot r >0$ and one has $\dot r<0$. For $r>3M$ the first corresponds to $\mathbb{T}_-(p)$ and the second corresponds to $\mathbb{T}_+(p)$. The roles are reversed for $2M<r<3M$. For $r=3M$ we have $\mathbb{T}_+(p)=\mathbb{T}_-(p)=(0,k_2,k_3)$. In~Figure \ref{fig:ssshadows} we depict these three cases for some fixed radii.
\begin{figure}[t!]
\centering
\centering
  \subfloat[\label{fig:ss3}]{%
    \includegraphics[width=.27\textwidth]{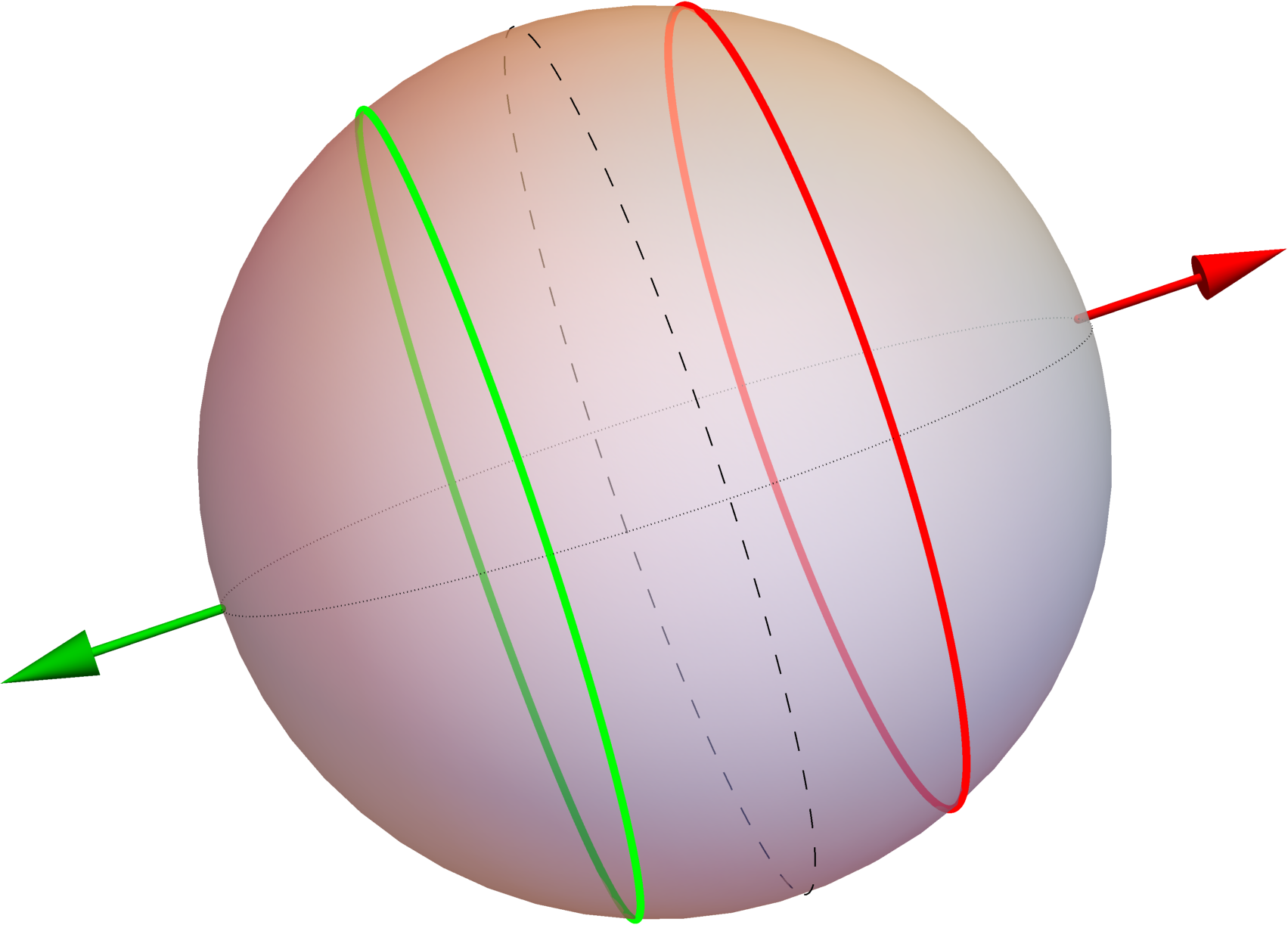}}\hfill
  \subfloat[\label{fig:ss2}]{%
    \includegraphics[width=.27\textwidth]{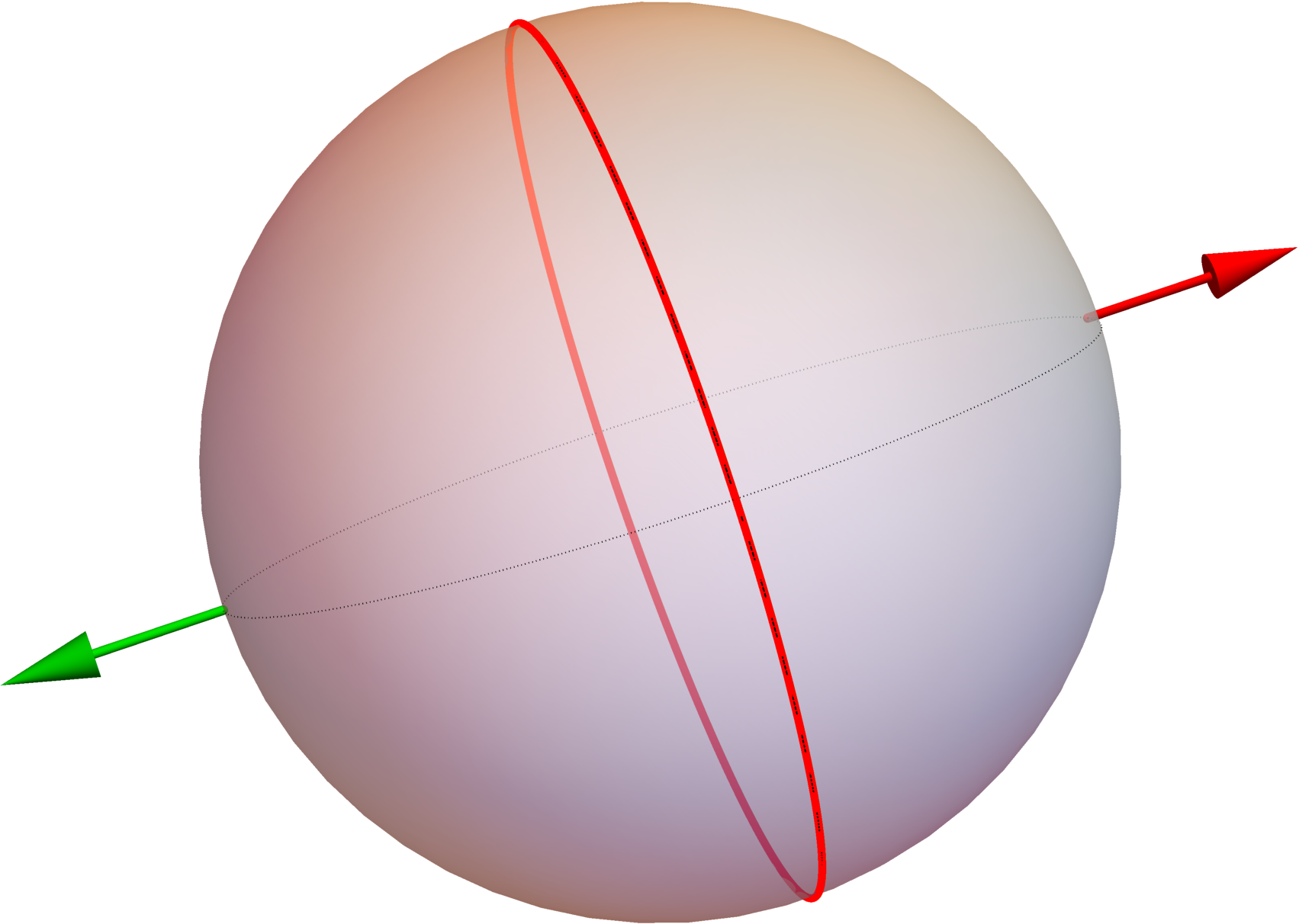}}\hfill
   \subfloat[\label{fig:ss1}]{%
    \includegraphics[width=.27\textwidth]{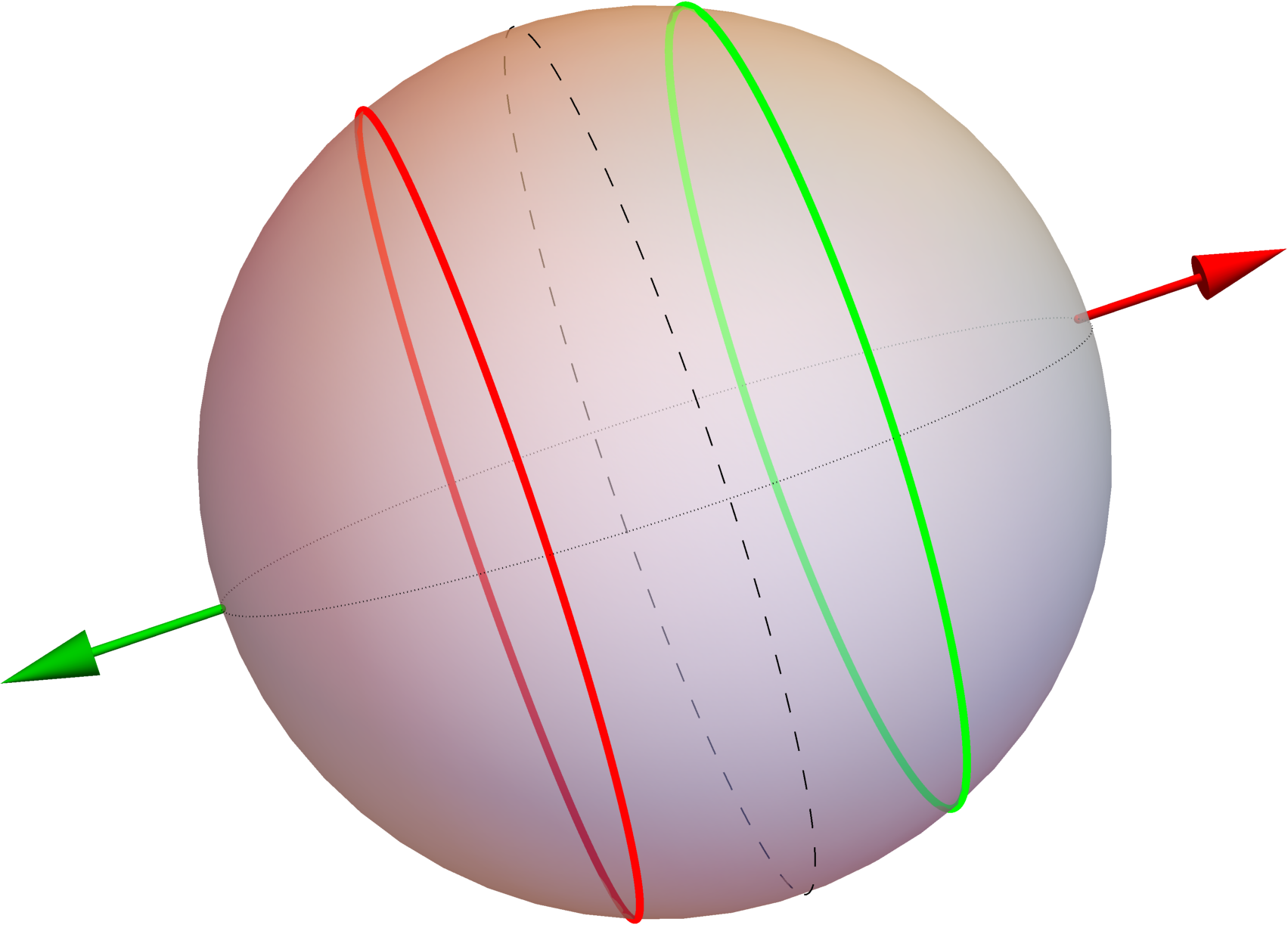}}\hfill
  \subfloat{%
    \includegraphics[width=.1\textwidth]{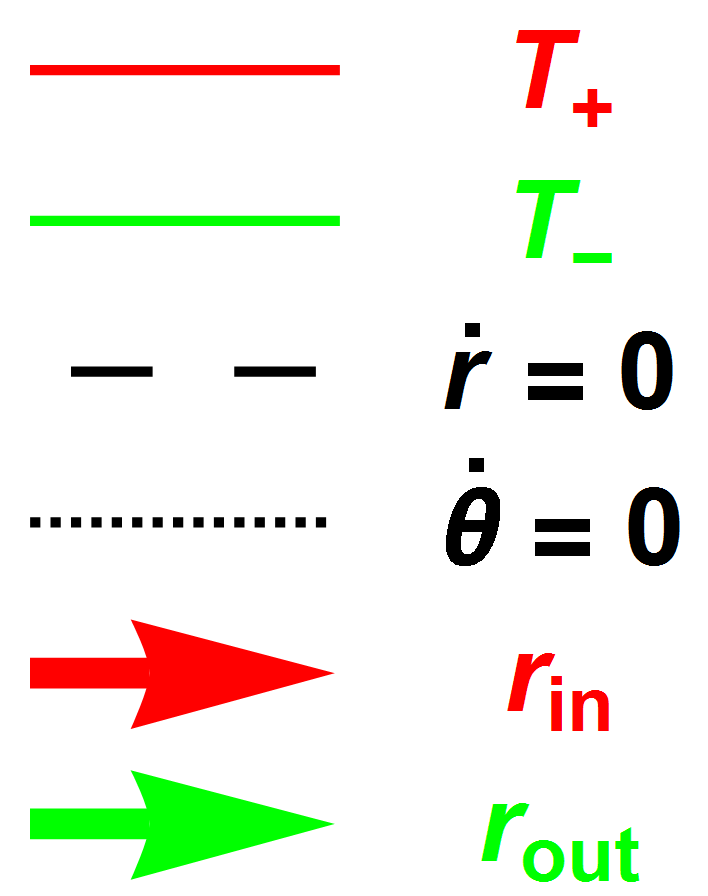}}\hfill
\caption{The trapped set on the celestial sphere of a standard observer at different radial location in a Schwarzschild DOC. Observer \protect\subref{fig:ss3} is located outside the photon sphere at $r=3.9M$, observer \protect\subref{fig:ss2} is located on the photon sphere at $r=3M$ and finally observer\protect\subref{fig:ss1} is located between the horizon and the photonsphere at $r=2.5M$. One can see that the future trapped set moves from the ingoing hemisphere in \protect\subref{fig:ss3} to the outgoing hemisphere in \protect\subref{fig:ss1} as one varies the location of the observer. The future and past trapped set coincide on the $\dot r =0$ line when the observer is located on the photon sphere at $r=3M$ in \protect\subref{fig:ss2} }
\label{fig:ssshadows}
\end{figure}
To conclude we see that $\mathbb{T}_+(p)$ and $\mathbb{T}_-(p)$ are circles on the celestial sphere independent of the value of $r(p)$. In \cite[p.14]{penrose_spinors_1987} it is shown that Lorentz transformations on the observer correspond to conformal transformations on the celestial sphere. They are equivalent up to M\"obius transformations on the complex plane when considering the $S_2$ to be the Riemann sphere. Hence circles on the celestial sphere stay circles under Lorentz transformations. As a consequence if $r(p_a)\neq r(p_b)$ then there exists a Lorentz transformation $(\mathbf{LT})$ such that  $\mathbb{T}_-(p_a)= \mathbf{LT}[\mathbb{T}_-(p_b)]$. This concept is sufficiently important that it deserves a proper definition.
\begin{definition}\label{def:shadowdeg}
The shadows at two points $p_1$, $p_2$ are called degenerate if, upon identification of the two celestial spheres by the orthonormal basis, there exists an element of the conformal group on $S^2$ that transforms $\mathbb{T}_-(p_1)$ into $\mathbb{T}_-(p_2)$.
\end{definition}

\begin{remark}
The shadow at two points $p_1$, $p_2$ being degenerate implies that for every observer at $p_1$ there exists an observer at $p_2$ for which the shadow on $S^2$ is identical. Because this notion compares structures on $S^2$, the two points do not have to be in the same manifold for their shadows to be degenerate. Just from the shadow alone an observer can not distinguish between these two configurations.
\end{remark}
The only reliably information an observer can obtain is thus, that  $\mathbb{T}_-(p)$ is a~proper circle on his celestial sphere.

\subsection{The trapped sets in Kerr}\label{sec:kerrshadow}
We will now discuss the properties of the sets  $\mathbb{T}_\pm(p)$ in Kerr. Note that the equations of motion for $r$ \eqref{eq:radial} and $\theta$ \eqref{eq:theta} have two solutions that differ only by a sign for a fixed combination of $E, L_z, K$. Therefore we know that the trapped sets will have a reflection symmetry across the $k_1=0$ and the $k_2=0$ planes. A sign change in $k_2$ maps both sets $\mathbb{T}_\pm(p)$ to themselves, while a sign flip in $k_1$ maps $\mathbb{T}_+(p)$ to $\mathbb{T}_-(p)$ and vice versa.\\
In the following we will use the parametrization of \cite{grenzebach_photon_2014}. We introduce the coordinates $\rho \in[0,\pi]$ and $\psi \in[0,2\pi)$ on the celestial sphere. Thus \eqref{eq:tangentplus} can be written as:
\begin{equation}\label{eq:tangentsph}
\dot\gamma( \rho,\psi)|_p =\alpha ( e_0+ \cos(\rho) e_1+ \sin(\rho) \cos(\psi) e_2+ \sin(\rho)\sin(\psi) e_3)
\end{equation}
The direction towards the black hole is given by $\rho=\pi$. Following \cite{grenzebach_photon_2014}  one finds the following parametrization of the celestial sphere in terms of constants of motion: 
\begin{subequations}\label{eq:comonsphere}
\begin{align}
\sin(\psi)&= \left.\frac{(\mathcal{L}-a)+a \cos^2 (\theta)}{\sqrt{\mathcal{K}}\sin(\theta)}\right|_p,\\
\sin(\rho)&= \left. \frac{\sqrt{\Delta \mathcal{K}}}{r^2 - a(\mathcal{L}-a)}\right|_p,
\end{align}
\end{subequations}
Analog to the functions \eqref{eq:E_trap} and \eqref{eq:Q_trap} which give the value of the conserved quantities in terms of the radius of trapping, we can give such relations for the conserved quotients $\mathcal{K}$ and $\mathcal{L}$ in this formulation they can be found for example in \cite{grenzebach_photon_2014}. To differ between the trapped radius and the observers radius we use $x\in[r_{min}(\theta),r_{max}(\theta)]$ to parametrize the conserved quantities of the trapped set. Here $r_{min}(\theta)$ and $r_{max}(\theta)$ are given as the intersection of a cone of constant $\theta$ with the boundary of the area of trapping. Note that it is part of the following proof to show that the given interval is in fact the correct domain for the parameter $x$. The parameter $x$ here corresponds to the radius of the trapped null geodesic which a particular future trapped direction is asymptoting to. We have then: 
\begin{subequations}\label{eq:altcomtrap}
\begin{align}
\mathcal{K}&= \frac{16 x^2 \Delta(x)}{(\Delta'(x))^2}\\
a( \mathcal{L}-a)&= \left( x^2-\frac{4x\Delta(x)}{\Delta'(x)}\right)
\end{align}
\end{subequations}
where $ \Delta'(x)=2 x - 2 M  $ is the partial derivative of $\Delta(x)$ with respect to $x$. Plugging \eqref{eq:altcomtrap} into \eqref{eq:comonsphere} we obtain:
\begin{subequations}\label{eq:parametrized}
\begin{align}
\sin(\psi)&= \frac{\Delta'(x)\{x^2 +a^2 \cos^2 (\theta(p))\}-4 x \Delta(x)}{4x\sqrt{\Delta(x)} a\sin(\theta(p))}\label{eq:psitox}\\
\sin(\rho)&= \frac{4x\sqrt{\Delta(r(p))\Delta(x)}}{\Delta'(x)(r(p)^2-x^2) + 4x \Delta(x)} := h(x)\label{eq:rhotox}
\end{align}
\end{subequations}
We are now ready to prove the following Theorem along the lines of \cite{paganini2018smoothness}.
\begin{thm}\label{thm:1}
The sets $\mathbb{T}_+(p)$ and $\mathbb{T}_-(p)$ are smooth curves on the celestial sphere of any timelike observer at any point in the exterior region of any subextremal Kerr spacetime. 
\end{thm}
\begin{proof}
 We start by anayzing the right hand side of \eqref{eq:psitox}:
\begin{equation}
    \begin{split}
\frac{d}{dx} &\left(\frac{\Delta'(x)\{x^2 +a^2 \cos^2 (\theta(p))\}-4 x \Delta(x)}{4x\sqrt{\Delta(x)}a\sin(\theta(p))}\right) \\
&\hphantom{\left(\frac{\Delta'(x)\{x^2 \}}{x}\right.}=\frac{\{x^2 + a^2 \cos^2(\theta(p))\}((M-x)^3-M(M^2-a^2))}{2x^2 \Delta(x)^{3/2} a \sin(\theta(p))},
 \end{split}
\end{equation}
which is strictly negative for $x\in(r_+,\infty)$. The limit of the right hand side of \eqref{eq:psitox} is given by $\infty$ for $x\rightarrow r_+$ and $-\infty$ for $x\rightarrow\infty$. Therefore the right hand side is invertible and $x(\sin(\psi))$ is a smooth function of $\psi$ with extrema at the extremal points of $\sin(\psi)$. As was shown in \cite{grenzebach_photon_2014} the minimum $x_{min}(\theta(p))$ at $\psi=\pi/2$ and the maximum of $x_{max}(\theta(p))$ at $\psi=3\pi/2$ correspond exactly to the intersections of a cone with constant $\theta$ with the boundary of the region of trapping. This can be seen by setting the left hand side of \eqref{eq:psitox} equal to $\pm1$, taking the square of the equation, solving for $\cos^2(\theta)$ and comparing to \eqref{eq:theta_turn}. Important here is that $[x_{min}(\theta(p)),x_{max}(\theta(p))]\subset [r_1,r_2]$ for all values of $\theta(p)$.
Now we take a look at the right hand side of equation \eqref{eq:rhotox}:
\begin{equation}
   \frac{d}{dx} \left(h(x)\right) =\frac{8(r(p)^2-x^2)\Delta(r(p))((x-M)^3+ M(M^2-a^2))}{\sqrt{\Delta(r(p))\Delta(x)}(4x \Delta(x)+ (r(p)^2-x^2)\Delta'(x))^2}.
\end{equation}
This is positive when $x<r(p)$ and negative when $x>r(p)$. The denominator never vanishes for $x\in(r_+,\infty)$ because: 
\begin{equation}
(4x\Delta(x)+ (r(p)^2-x^2)\Delta'(x))|_{\{r(p)=r_+, x=r_+\}}=0
\end{equation}
and 
\begin{align}
\label{eq:dx}
\frac{d}{dx}(4x\Delta(x)+ (r(p)^2-x^2)\Delta'(x))&=2(3x^2-6Mx+2a^2+r(p)^2)>0\\
\frac{d}{dr(p)}(4x\Delta(x)+ (r(p)^2-x^2)\Delta'(x))&=2r(p)\Delta'(x)>0,
\end{align}
where we used $r(p)>r_+>M>a$ in \eqref{eq:dx}. \\
If we set $x=r(p)$ in \eqref{eq:rhotox} then the right hand side is equal to $1$. Furthermore in any of the limits $r(p)\rightarrow r_+$, $r(p)\rightarrow \infty$, $x\rightarrow r_+$, and $x\rightarrow\infty$ it goes to zero. \begin{case}If $p\notin \mathcal{A}$ hence if $r(p)\notin[x_{min}(\theta(p)),x_{max}(\theta(p))]$ then the two functions:
\begin{align}
\rho_1(\psi)=\arcsin(h(x(\psi)))&: [0,2\pi)\rightarrow [\rho_{1_{min}},\rho_{1_{max}}]\subset \left(0,\frac{\pi}{2}\right)\\
\rho_2(\psi)=\pi-\arcsin(h(x(\psi)))&: [0,2\pi)\rightarrow[\rho_{2_{min}},\rho_{2_{max}}]\subset\left(\frac{\pi}{2},\pi\right)
\end{align}
are both smooth with $\rho_1(0)=\rho_1(2\pi)$ and $\rho_2(0)=\rho_2(2\pi)$. If $p$ is between the region of trapping and the asymptotically flat end, the function $\rho_2(\psi)$ corresponds to $\mathbb{T}_+(p)$ and $\rho_1(\psi)$ corresponds to $\mathbb{T}_-(p)$. Because $(\pi/2,\pi]$ corresponds to the geodesic with $\dot r< 0$. If $p$ is between the region of trapping and the horizon then the role of $\rho_1(\psi)$ and $\rho_2(\psi)$ are switched.
\end{case}
\begin{case}If $p\in \mathcal{A}$ we need to do some extra work. For simplicity we only consider the interval $\psi\in[\pi/2,3\pi/2]$ as the rest follows by symmetry of $\sin(\psi)$ in $[0,\pi]$ across $\pi/2$ and in $[\pi,2\pi]$ across $3\pi/2$. Let $\arcsin(x)$ map into this interval, then we define:
\begin{equation}
    \psi_0(r(p))= \arcsin\left(\frac{\Delta'(r(p))\{r(p)^2 +a^2 \cos^2 (\theta(p))\}-4 r(p) \Delta(r(p))}{4r(p)\sqrt{\Delta(r(p))}a\sin(\theta(p))}\right).
\end{equation}
The two functions:
\begin{align}
\rho_3(\psi)&=\begin{cases}
\arcsin(h(x(\psi))) &\text{ if } \psi\in[\pi/2,\psi_0(r(p))]\\
\pi-\arcsin(h(x(\psi)))&\text{ if } \psi\in (\psi_0(r(p)),3\pi/2]
\end{cases}\\
\rho_4(\psi)&=\begin{cases}
\pi-\arcsin(h(x(\psi))) &\text{ if } \psi\in[\pi/2,\psi_0(r(p))]\\
\arcsin(h(x(\psi)))&\text{ if } \psi\in (\psi_0(r(p)),3\pi/2]
\end{cases}
\end{align}
are then smooth on $[\pi/2,3\pi/2]$. For a proof see Appendix \ref{app:A} and note that at $\psi_0$, $h(x(\psi))$ satisfies the conditions required in the appendix. Since $p\in \mathcal{A}$ we have that $x_{min}(\theta(p))<r(p)<x_{max}(\theta(p))$. Therefore the geodesic on the celestial sphere parametrized by $x_{max}(\theta(p))$ has to have $\dot r>0$ and thus has to be in $[0,\pi/2)$.  On the other hand the geodesic on the celestial sphere parametrized by $x_{min}(\theta(p))$ has to have $\dot r<0$ and thus has to be in $(\pi/2,\pi]$. In fact by the monotonicity of the right hand side of \eqref{eq:psitox} and the fact that $x(\psi_0)=r(p)$ we know that for $\psi \in [\pi/2,\psi_0)$ we have $x(\psi)<r(p)$ and for $\psi \in (\psi_0,3\pi/2]$ we have $x(\psi)>r(p)$. Thus we can conclude that for $p\in \mathcal{A}$, $\rho_4$ corresponds to $\mathbb{T}_+(p) $ and $\rho_3$ corresponds to $ \mathbb{T}_-(p)$ and thus both sets are smooth. 
\end{case}
\begin{case}In the special case when $r(p)= x_{max}(\theta(p))$ or $r(p)= x_{min}(\theta(p))$ the functions $\rho_1$ and $\rho_2$ describing $\mathbb{T}_\pm(p)$ do reach $\rho=\pi/2$ at $\psi=3\pi/2$ or $\psi=\pi/2$ respectively. However since in these cases we have that:
\begin{equation}
    \frac{d^2}{d\psi^2}(h(x(\psi)))=0,
\end{equation}
the two sets meet at this point tangentially and do not cross over into the other hemisphere.
\end{case}
This concludes the proof. 
\end{proof}
\begin{remark}
In \cite{grenzebach_photon_2014} it was observed that $\rho_{max}$ of $\mathbb{T}_+(p) $ always corresponds to the trapped geodesic with $x_{min}(\theta(p))$ and $\rho_{min}$ of $\mathbb{T}_+(p) $ always corresponds to the trapped geodesic with $x_{max}(\theta(p))$ . When $p$ is outside the region of trapping $h(x)|_{x_{max}}$ is a~local maximum of $h(x(\psi))$ (as a function of $\psi$) and $h(x)|_{x_{min}}$ is a~local minimum of $h(x(\psi))$. When $p$ is between the region of trapping and the horizon $h(x)|_{x_{max}}$ is a~local minimum of $h(x(\psi))$ and  $h(x)|_{x_{min}}$ is a local maximum of $h(x(\psi))$. Since outside $\mathbb{T}_+(p) $ is always described by $\rho_2(\psi)$ and inside by $\rho_1(\psi)$, $\rho_{min}$ then always corresponds to $x_{min}$ and $\rho_{max}$ always corresponds to $x_{max}$. This also holds for $p\in\mathcal{A}$. \\
This observation means that the null geodesic approaching the innermost photon orbit has the smallest impact parameter (deviation from the radially ingoing null geodesic) and the null geodesic approaching the outermost photon orbit has the largest impact parameter.\\
For $\mathbb{T}_-(p) $ the correspondence is switched.
\end{remark}
\begin{remark}
The parametrization for $\sin(\psi)$ breaks down on the rotation axis, the one for $\sin(\rho)$ however remains valid with only one possible value for $x=r_3$. Due to the symmetry at these points we know that $\mathbb{T}_\pm(p)$ are described by proper circles on the celestial sphere. The PND is aligned with the axis of the rotation symmetry and hence the sets are symmetric under rotation along $\psi$. The situation is therefore equivalent to Schwarzschild and an observer can not distinguish whether they observes a Schwarzschild black hole or a Kerr black hole in the direction of the rotation axis. 
\end{remark}
\begin{remark}
We have only proved Theorem \ref{thm:1} for one standard observer at any particular point. However since any other observer at this point is related to the standard observer by a Lorentz transformation and the Lorentz transformations act as conformal transformations on the celestial sphere \cite[p.14]{penrose_spinors_1987}, the Theorem indeed holds for any observer. In \cite{grenzebach_aberrational_2015} the quantitative effect on the shape of the shadow of boosts in different directions are discussed.
\end{remark}
\begin{remark}
The parametrization for the trapped set on the celestial sphere of any standard observer in \cite{grenzebach_photon_2014,grenzebach_photon_2015} was derived for a much more general class of spacetimes. Therefore  Theorem \ref{thm:1} might actually hold for these cases as well. However, this is beyond the scope of this lecture note.
\end{remark}
From Theorem \ref{thm:1} we immediately get the following Corollary:
\begin{cor}
For any observer at any point $p$ in the exterior region of a subextremal Kerr spacetime we have that for any $k\in \mathbb{T}_+(p)$ and any $\epsilon>0$
\begin{itemize}[noitemsep]
\item $B_\epsilon(k)\cap \Omega_{\mathcal{H}^+}(p)\neq \emptyset$
\item $B_\epsilon(k)\cap \Omega_{\scri^+}(p)\neq \emptyset$.
\end{itemize}
\end{cor}
\noindent Thus if we interpret the celestial sphere as the initial data space for null geodesics starting at $p$, the above Corollary is a coordinate independent formulation of the fact that trapping in the exterior region of subextremal Kerr black holes is unstable.\\
See Figure \ref{fig:kerrshadows} as an example of how the trapped sets change under a variation of the radial location of the observer.\\
\begin{figure}[t!]
\centering
   \subfloat[\label{fig:kerr1}]{%
    \includegraphics[width=.3\textwidth]{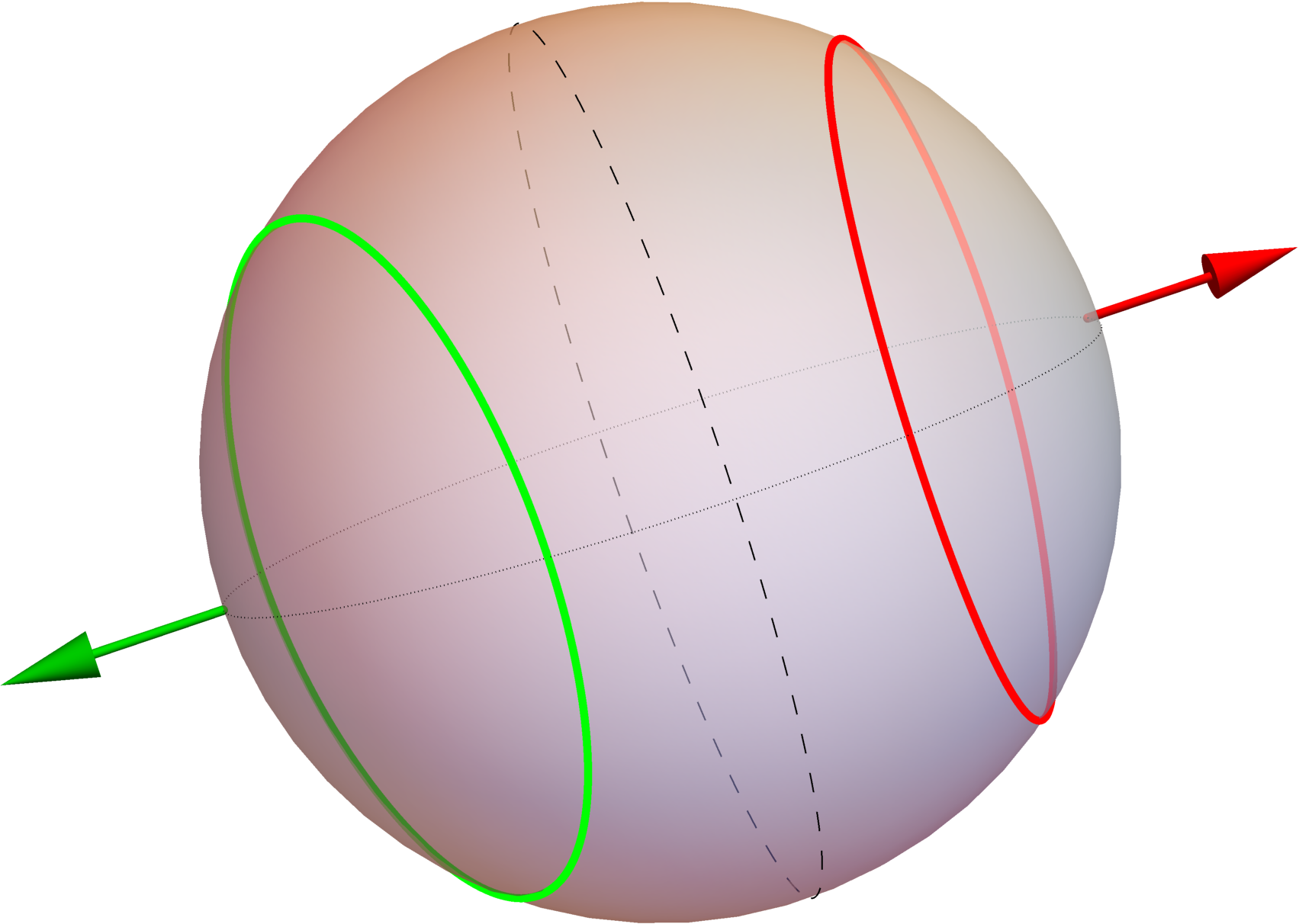}}\hfill
  \subfloat[\label{fig:kerr2}]{%
    \includegraphics[width=.3\textwidth]{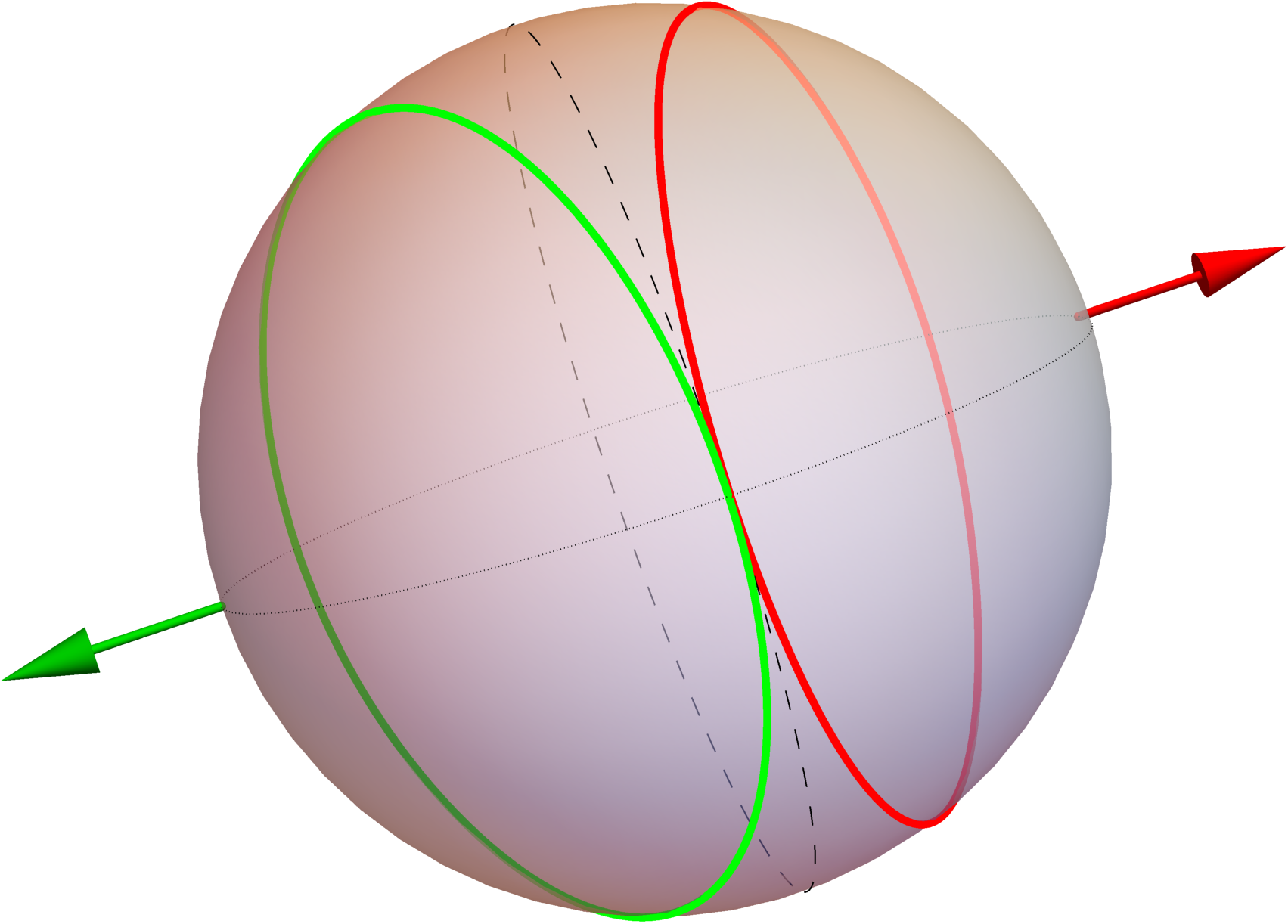}}\hfill
  \subfloat[\label{fig:kerr3}]{%
    \includegraphics[width=.3\textwidth]{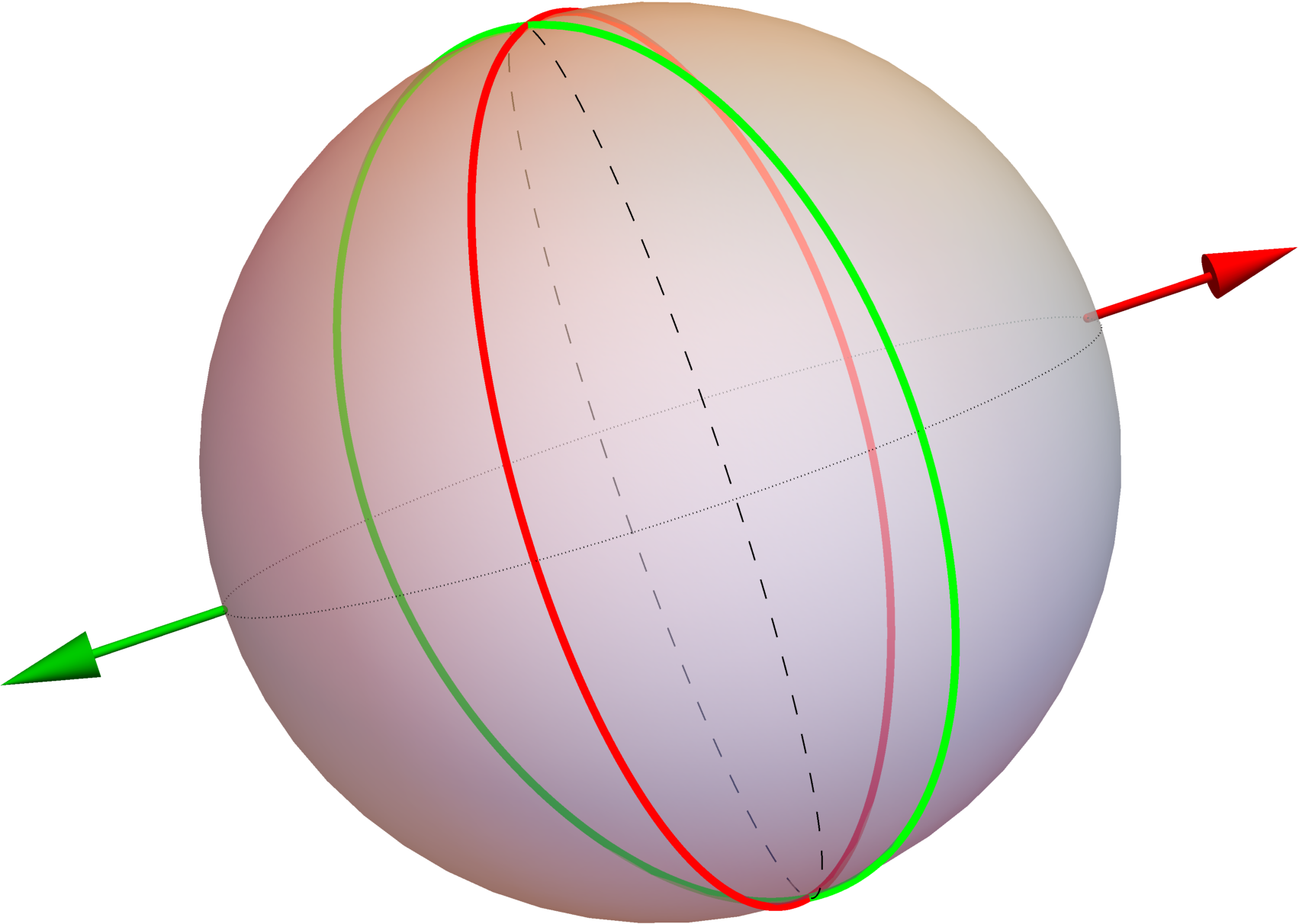}}\hfill\\
\hspace*{\fill} \subfloat[\label{fig:kerr4}]{%
    \includegraphics[width=.3\textwidth]{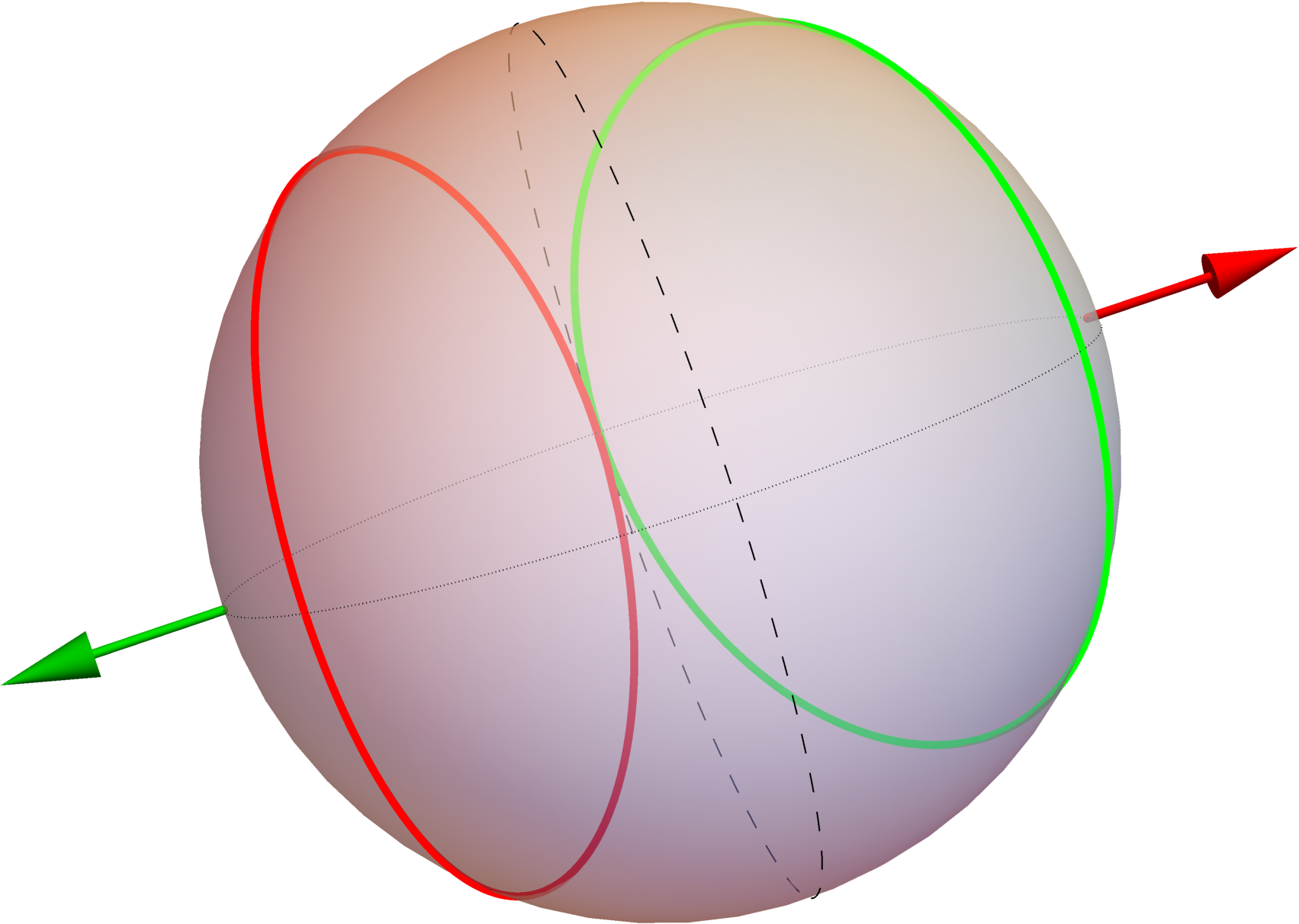}}\hfill
  \subfloat[\label{fig:kerr5}]{%
    \includegraphics[width=.3\textwidth]{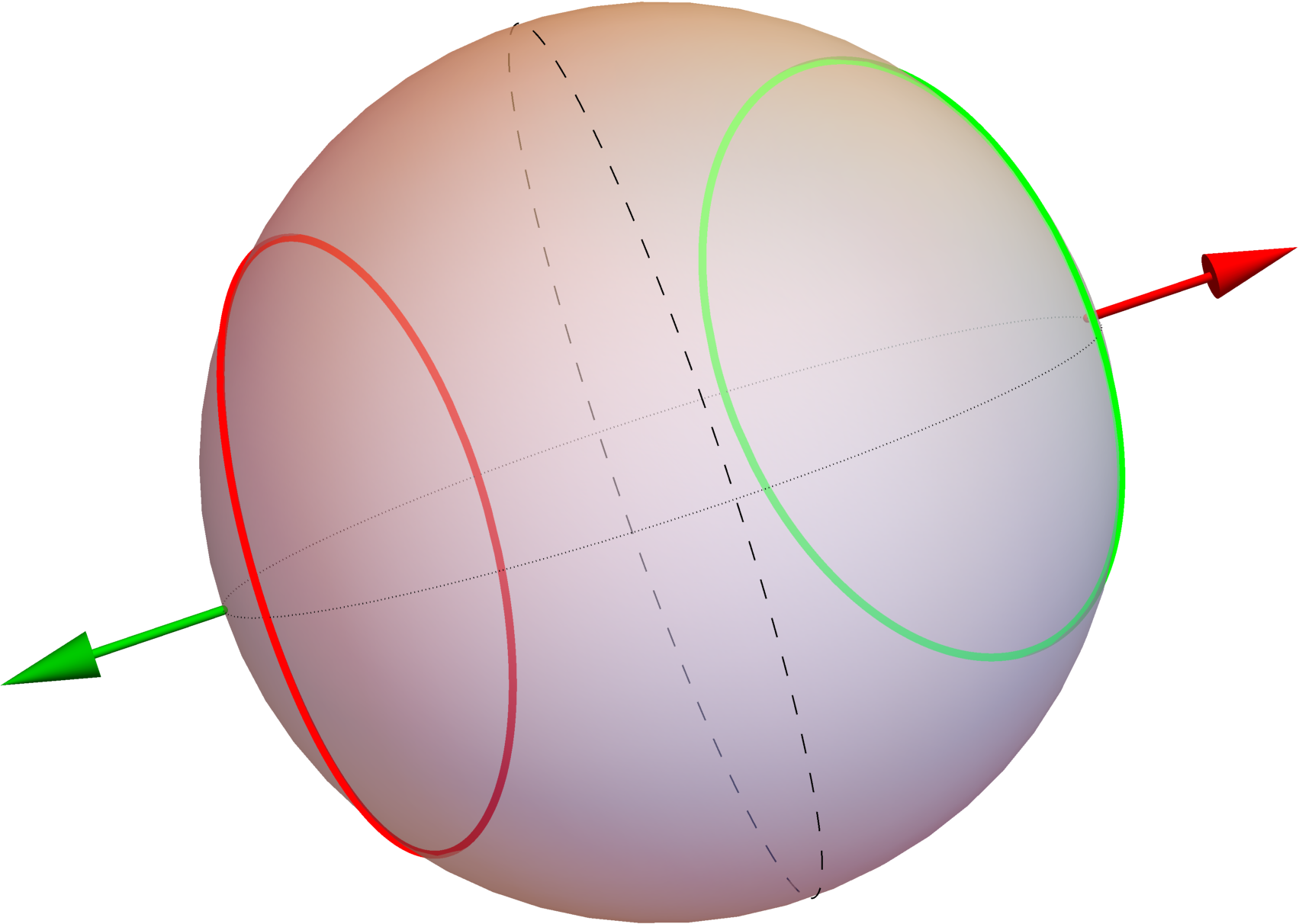}}\hfill
    \subfloat{%
    \includegraphics[width=.1\textwidth]{legend.png}}\hspace*{\fill}
  \caption{The trapped set on the celestial sphere of a standard observer at different radial locations in the equatorial plane of the exterior region of a Kerr black hole with $a=0.9$. Observer \protect\subref{fig:kerr1} is located outside the region of trapping at $r=5M$. Observer \protect\subref{fig:kerr2} is located on the outer boundary of the region of trapping at $r=r_2=3.535M$. Observer \protect\subref{fig:kerr3} is located inside the region of trapping at $r=r_3=2.56M$. Observer  \protect\subref{fig:kerr4} is located on the inner boundary of the area of trapping at $r=r_1=1.73M$ Finally observer \protect\subref{fig:kerr5} is located between the horizon and the area of trapping at $r=1.59M$. Again one can observe how the two trapped sets move in opposite directions on the celestial sphere as the observer approaches the black hole. In \protect\subref{fig:kerr1} the future trapped set is on the ingoing hemisphere and the past trapped set is on the outgoing hemispere. In \protect\subref{fig:kerr2} they meet in one point tangentially but are still entirely in one hemisphere except for that one point. In \protect\subref{fig:kerr1} the trapped sets intersect in two points and both have parts in both hemispheres. In \protect\subref{fig:kerr4} they only meet at one point tangentially again (now on the "other" side of the celestial sphere) and finally in \protect\subref{fig:kerr5} the future trapped set is entirely in the outgoing hemisphere and the past trapped set is entirely in the ingoing hemisphere.}
\label{fig:kerrshadows}
\end{figure}
Even though the qualitative features of $\mathbb{T}_\pm(p)$ do not change under a change of parameters, the quantitative features do. In \cite{hioki_measurement_2009,li_measuring_2014} it is discussed what information can be read of from the shadow at infinity. For points inside the manifold as considered in this work, this question was resolved in \cite{mars2017fingerprints}. It has been shown that in principal an observer away from the rotation axis can obtain information about the angular momentum $a$ of the black hole, their radial distance $r$ and inclination above the equatorial plane $\theta$ as well as their state of motion as compared to a~standard observer given by $e_0$ of \ref{eq:tetrad}. In the following, as an illustration of that result, we will present numerical results that show that the radial degeneracy of the shadow is broken in Kerr spacetimes away from the symmetry axis for. 
We take the stereographic projection: \cite[p.10]{penrose_spinors_1987}
\begin{subequations}\label{eq:projection}
\begin{align}
\zeta&= \frac{x+ i y}{1-z}\\
x&= \sin(\rho) \sin(\psi) \qquad y= \sin(\rho) \cos(\psi) \qquad z= \cos(\rho)
\end{align}
\end{subequations} 
of $\mathbb{T}_-(p)$ in order to compare the exact shape of the shadow for different observers. Here the values of $\rho$ and $\psi$ are given by the parametrization in \eqref{eq:altcomtrap}. The above projection for the celestial sphere of a standard observer in the exterior region of a Kerr black hole  guarantees that projection of the shadow in the complex plane will have a reflection symmetry about the real axis. This is due to the symmetry of the shadow on the celestial sphere of a standard observer under a sign flip in the $k_2$ component. To compare conformal classes of shadows for observers at different points of different black hole spacetimes, we establish the notion of a canonical observer together with a canonical projection. This will allow us to compare the shape of the shadow for these observers directly. The canonical observer together with the canonical projection are defined in such a way that the points on $\mathbb{T}_-(p)$ with $\Psi =\pi/2,3\pi/2$ will correspond to the points $(1, 0)$ and $(-1, 0)$ in the complex plane. On the practical side we take the stereographic projection along the radially ingoing direction for the standard observer and then apply conformal transformations to the projection of the shadow (translation and rescaling). In the following $C_{S^2}$ denotes the union of all conformal transformations on $S^2$.
\begin{remark}
In order to prove that the conformal class of two curves $x$ and $y$ on $S^2$ are equal it is sufficient to find an $x_0\in C_{S^2}[x]$ and a $y_0\in C_{S^2}[y]$ such that $x_0=y_0$. By choosing a canonical observer, we choose a fix representative $c_0[p]\in C_{S^2}[\mathbb{T}_-(p)]$ of each conformal class and thereby eliminate the freedom of performing conformal transformations on the celestial sphere of an observer. Hence if the shape of the canonical representative of any two shadows $c_0[p_1(a_1,\theta_1,r_1)]$, $c_0[p_2(a_2,\theta_2,r_2)]$ coincide, then their conformal class is equal. On the other hand if $c_0[p_1(a_1,\theta_1,r_1)]\neq c_0[p_2(a_2,\theta_2,r_2)]$ then the conformal classes of the shadow at $p_1$ and that at $p_2$ are different.
\end{remark}
 The breaking of the radial degeneracy in Kerr spacetimes can be seen in Figure \ref{fig:radialdeg}, where we plotted two black hole shadows, for canonical observers located in a Kerr spacetime with $a = 0.99$, at $ \theta = \pi/2 $, and different $r$ values: $r_1 = 5, r_2 = 50$. These particular values were chosen because for this case the degeneracy breaking is visible by the naked eye from the plot in Figure \ref{fig:radialdeg}. In general the breaking of the radial degeneracy is hardly visible in the plot. The breaking of the radial degeneracy implies that in principle when observing the shadow of a black hole, as has been done by the Event Horizon Telescope, one would have to take the distance from said black hole into consideration when trying to extract the black holes parameters from the observation. However the influence of the radial degeneracy for far away observers is quite likely a lot smaller than the resolution that can be achieved from such observations. 
\begin{figure}[t!]
  \subfloat{%
    \includegraphics[width=.8\textwidth]{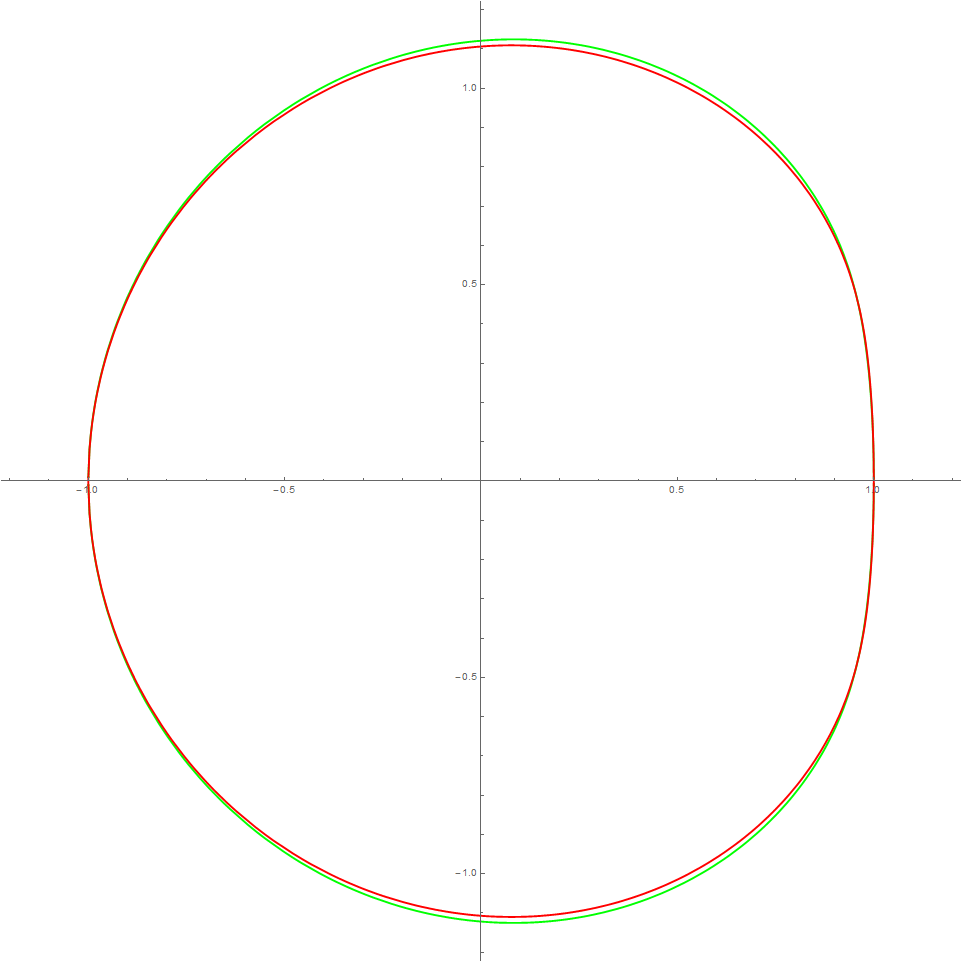}}\hfill
  \subfloat{%
    \includegraphics[width=.1\textwidth]{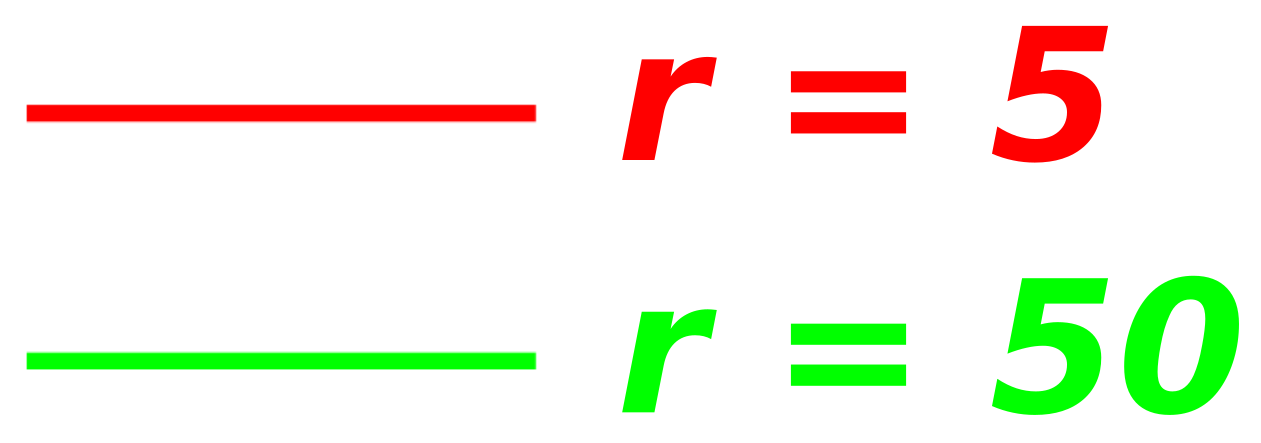}}\hfill
\caption{Breaking of the radial degeneracy for the shadow of a~Kerr black hole. Two canonical observers with $a = 0.99$, at~$\theta = \pi/2$, for two different values $r = 5$ and $r=50$.    }
\label{fig:radialdeg}
\end{figure}

\subsection{The Celestial Sphere as a Tool}\label{sec:cstool}
We will now take a look at some questions related to trapping beyond the Kerr family of spacetimes. 
In this section we follow closely \cite[p.72]{paganini2018role}. We present partial results towards resolving the question whether there can exist trapped null geodesics orthogonal to the Killing vector field $T$ in general smooth stationary black hole space times with positive surface gravity. The results in this section do not allow a conclusive answer, but they might turn out to be a small step on the way towards a full proof.

In the following we will make heavy use of the continuous dependence on initial data for the geodesic equation. We will use the following theorem.\footnote{Despite the fact that this is a rather basic statement, we were unable to find a source for this Theorem in the form we intend to use here. Nevertheless it is to be expected that this result is in fact well known.} 

\begin{thm}\label{thm:continuity}
Let $\mathcal M$ be a $C^1$ spacetime. Let $p$ and $q$ be two points on the same geodesic $\gamma$, with the tangent vector at $p$ pointing towards $q$. Let $p$ and $q$  be separated by a finite affine parameter. Let $\mathcal U (q) $ be an open neighbourhood of $q$ in $\mathcal{M}$. Then there exists an open neighbourhood $T\mathcal{U}(p, \dot \gamma|_p)$ of $(p, \dot \gamma|_p)$ in $T\mathcal M$ such that any geodesic in this neighbourhood intersects $\mathcal{U}(q)$.
\end{thm}
\begin{proof}
Let $\mathcal{U}_n$, $n\in[0,N]$ be a finite sequence of normal neighbourhoods that cover $\gamma$ between $p$ and $q$ with $p\in \mathcal{U}_0$ and $q\in \mathcal{U}_N$. Let $\Psi_n$ be the coordinate chart that belongs to the normal neighbourhood $\mathcal{U}_n$. \\
Now we pick a sequence of points $p_n$ such that $p_n\in \mathcal{U}_n\cap \mathcal{U}_{n-1}$.\\
For $p_N$ the theorem is true by the continuous dependence on initial data for ODEs in $\Reals^n$, for a proof see \cite[p.95]{hartman2002ordinary}.\\
For $p_{N-1}$ we have to do a little extra work. We have to show that there exists a neighbourhood $T\mathcal{U}(p_{N-1}, \dot \gamma|_{p_{N-1}})$ such that for any $\tilde\gamma(\lambda)$ that is a solution to the geodesic equation with initial data in $T\mathcal{U}(p_{N-1}, \dot \gamma|_{p_{N-1}}) $ there exists a parameter $\lambda_0$ such that $(\tilde\gamma(\lambda_0),\dot\tilde\gamma (\lambda_0))\in T\mathcal{U}(p_{N}, \dot \gamma|_{p_{N}})$. This is true by applying the continuous dependence on initial data for ODEs in  $\mathcal{U}_{N-1}$ and observing that for points in $\mathcal{U}_N\cap \mathcal{U}_{N-1}$ the determinant of the Jacobian of the map $\Psi_N \circ (\Psi_{N-1})^{-1}$ from $\Reals^n$ to $\Reals^n$ is bounded from above and below.\\
Now if it is true for $p_{N-1}$ it is clear that it is true for any $p_n$ and by that also for $p$ itself.
\end{proof}

For our further argument, we will only need the following corollary of Theorem \ref{thm:continuity}.

\begin{cor}\label{cor:geodcont}
Let $\mathcal M$ be a $C^1$ spacetime. Let $p$ and $q$ be two points on the same null geodesic with the tangent vector $\dot \gamma(k|_p)$ at $p$ pointing towards $q$. Let $p$ and $q$ be separated by a finite affine parameter. Let $e_0$ be a unit timelike vector and $S^2(e_0)$ the associated celestial sphere. Let $\mathcal U (q) $ be an open neighbourhood of $q$ in $\mathcal{M}$. Then there exists an open neighbourhood $B_\epsilon(k)$ of $k$ on $S^2(e_0)$ such that any null geodesic $\gamma(\tilde k|_p)$ with $\tilde k \in B_\epsilon(k)$ intersects $\mathcal{U}(q)$.
\end{cor}
Where $\gamma(\tilde k|_p)$ is given by definition \ref{def:gammak}.

\subsubsection{Existence of Trapping in General Black Hole Spacetimes}
In the following, we~will show that the existence of trapping is a generic feature of black hole spacetimes. 
Recall that $\HH^\pm=\overline{J^-(\scri^+)\cap J^+(\scri^-)} \backslash J^\mp(\scri^\pm) $. In Kerr, Schwarzschild and Minkowski $\scri^\pm$ are smooth. The various stability proofs for Minkowski space \cite{christodoulou_global_1993,hintz2017global,bieri2009extension} come to different conclusions on the regularity of $\scri^+$ in these more generic settings depending on the choice of initial data. For the arguments in this section it would of course be easiest to assume that $\scri^\pm$ are smooth, however for the simple argument presented here it is sufficient to assume that they are in $C^1$. 
\begin{lemma}\label{lem:OHopen}
Let $\mathcal M$ be a $C^1$ spacetime. Let $\mathcal{M}$ be compactifiable with complete $\scri^\pm$.  Let $p$ be any point in $J^-(\scri^+)\cap J^+(\scri^-)$. If the sets $\Omega_{\HH^\pm}(p)$ are non empty then they are open in $S^2(e_0)$ 
\end{lemma}
\begin{proof}
Let $k$ be in $\Omega_{\HH^+}(p)$. Let $q_0$ be the point where $\gamma(k|_p)$ intersects the event horizon $\mathcal{\HH^+}$. Since $q_0$ is a regular point in $\mathcal M$, $\gamma(k|_p)$ can be extended across $q_0$ and hence the affine parameter between $p$ and $q_0$ has to be finite. Let $q$ be a point on $\gamma(k|_p)$ in the future of $q_0$ and $\mathcal{U}(q)$ an open neighbourhood of $q$ in $\mathcal M$ with $\mathcal{U}(q)\cap \overline{J^-(\scri^+)\cap J^+(\scri^-)}=\emptyset$ then Corollary \ref{cor:geodcont} guarantees that the Lemma is true, because any null geodesic that intersects $\mathcal{U}(q)$ has to intersect the horizon.
\end{proof}
The proof for $\Omega_{\HH^-}$ is identical.

\begin{lemma}\label{lem:OIopen}
Let $\mathcal M$ be a $C^1$ spacetime. Let $\mathcal{M}$ be compactifiable with complete $\scri^\pm$ in $C^1$.  Let $p$ be any point in $J^-(\scri^+)\cap J^+(\scri^-)$. Then $\Omega_{\scri^\pm}(p)$ are open sets in $S^2$.
\end{lemma}
\begin{proof}
Let $k$ be in $\Omega_{\scri^+}(p)$. Let $\tilde M$ be the closure of the compactification of $M$. Let $q_0$ be the point where $\gamma(k|_p)$ intersects $\scri^+$. The spacetime and therefore also the null geodesic $\gamma(k|_p)$ can be extended across the conformal boundary of $\tilde M$. Note that for the following it is not relevant that this extension is not unique but it has to be in $C^1$. Now $q_0$ is a regular point in this extension and the affine parameter between $p$ and $q_0$ is finite in the compactification. Let $q$ be a point on $\gamma(k|_p)$ in the future of $q_0$ and $\mathcal{U}(q)$ an open neighbourhood of $q$ in the extension of $\tilde M$ with $\mathcal{U}(q)\cap \overline{J^-(\scri^+)\cap J^+(\scri^-)}=\emptyset$ then Corollary \ref{cor:geodcont} guarantees that the Lemma is true, because any null geodesic that intersects $\mathcal{U}(q)$ has to intersect $\scri^+$.
\end{proof}
The proof for $\Omega_{\scri^-}$ is identical. Note that the regularity assumptions on $\scri^\pm$ in this Lemma can quite likely be relaxed and replaced by a sufficiently fast fall off in the requirements of asymptotic flatness. The argument would go along the lines that in Minkowski space null geodesics can only have outward turning points. For Schwarzschild and Kerr this is true far enough from the black hole. This is expected to be true far enough out for any asymptotically flat spacetime. Therefore once the outgoing null geodesic enters this region it can only move further away. This will be true for the ones close by as well. Thereby establishing the openness of the sets without going all the way out to $\scri^\pm$ and therefore independently of the fact whether the manifold can be extended across $\scri^\pm$ in any regularity.\\
These Lemmas allow us to prove the main theorem in this section. 

\begin{thm}
Let $\mathcal M$ be a $C^1$ spacetime. Let $\mathcal{M}$ be compactifiable with complete $\scri^\pm$ in $C^1$. Let $p$ be any point in $J^-(\scri^+)\cap J^+(\scri^-)$. If $\Omega_{\HH^\pm}(p)$ is non-empty then 
\begin{itemize}
    \item $\mathbf T _\pm(p)$ is non-empty
    \item let $w(\lambda)$ be any continuous path on $S^2(p)$ with  $\lambda \in [0,1]$ such that $w(0)\in \Omega_{\HH^\pm}(p)$ and $w(1)\in \Omega_{\scri^\pm}(p)$, we have that $w(\lambda)\cap \mathbf T _\pm(p)\neq \emptyset$. 
\end{itemize}
\end{thm}
\begin{proof}By completeness of $\scri^\pm$ and the fact that $p\in J^-(\scri^+)\cap J^+(\scri^-)$ we have that $\Omega_{\scri^\pm}(p)$ are non-empty.
By definition we have $\Omega_{\HH^\pm}(p)\cap\Omega_{\scri^\pm}(p)=\emptyset$. By Lemma \ref{lem:OHopen} and Lemma \ref{lem:OIopen} the sets are open and the theorem follows directly.
\end{proof}

\subsubsection{The Celestial Sphere and \texorpdfstring{$\mathbf T$}{T}-Orthogonal Trapping}
Again following \cite{paganini2018role}, we~will show that under the assumption that trapping is unstable we can prove that no trapped null geodesics with negative energy can exist in any stationary black hole spacetime. It is conceivable that this result can be strengthened to show that instability of trapping actually implies that all trapped null geodesics have to have positive energy. However at present times this remains an open problem. See the informal lecture notes \cite{claudio} for more details. 

\begin{lemma}
Let $\mathcal{M}$ be a smooth stationary black hole spacetime with Killing vector field $T$ with one asymptotically flat end. Further assume future and past trapping to be unstable in the sense that for any observer at any point $p$ in the exterior region we have that for any $k\in \mathbb{T}_+(p)$ for any $\epsilon>0$
\begin{itemize}
\item $B_\epsilon(k)\cap \Omega_{\mathcal{H}^+}(p)\neq \emptyset$
\item $B_\epsilon(k)\cap \Omega_{\scri^+}(p)\neq \emptyset$.
\end{itemize}
then any trapped null geodesic $\gamma$ in the exterior region satisfies $g(T,\dot \gamma)\geq0$. 
\end{lemma}
\begin{proof}
Note that for every $k\in \Omega_{\mathcal{I}^+}(p)$ we have $-\dot\gamma^\mu( k|_p) T_\mu>0$. 
Due to the instability condition we can choose a convergent sequence $q_i \in \Omega_{\mathcal{I}^+}(p)$ with $\lim_{i \to \infty} q_i =k$ for any  $k\in \mathbb{T}_+(p)$ .
 We then have that 
 \begin{equation}
 E(k)=-\dot\gamma^\mu( k|_p) T_\mu =\lim_{i \to \infty} -\dot\gamma^\mu( q_i|_p) T_\mu \geq 0
 \end{equation}
The statement then follows from the fact that the trapped set is given by $\mathbb{T}(p):= \mathbb{T}_+(p)\cap\mathbb{T}_-(p)$.
\end{proof}
This concludes our excursion beyond the Kerr family of spacetimes. 

\FloatBarrier
\section{Application}\label{sec:app}
Everything we have derived about the behaviour of null geodesics in Kerr spacetimes can be represented in a couple of simple plots. See Figure \ref{fig:potentials}, Figure \ref{fig:ssshadows} and Figure \ref{fig:kerrshadows} as examples; in the Mathematica notebook provided with this paper \cite{blazej_marius} the parameters $a/M$ and $\QQ$ as well as the location of the observer $\{r(p),\theta(p)\}$ can be varied. This allows one to develop an intuitive understanding of the influence of these parameters.

\begin{figure}[ht!]
\centering
\includegraphics[width=120mm]{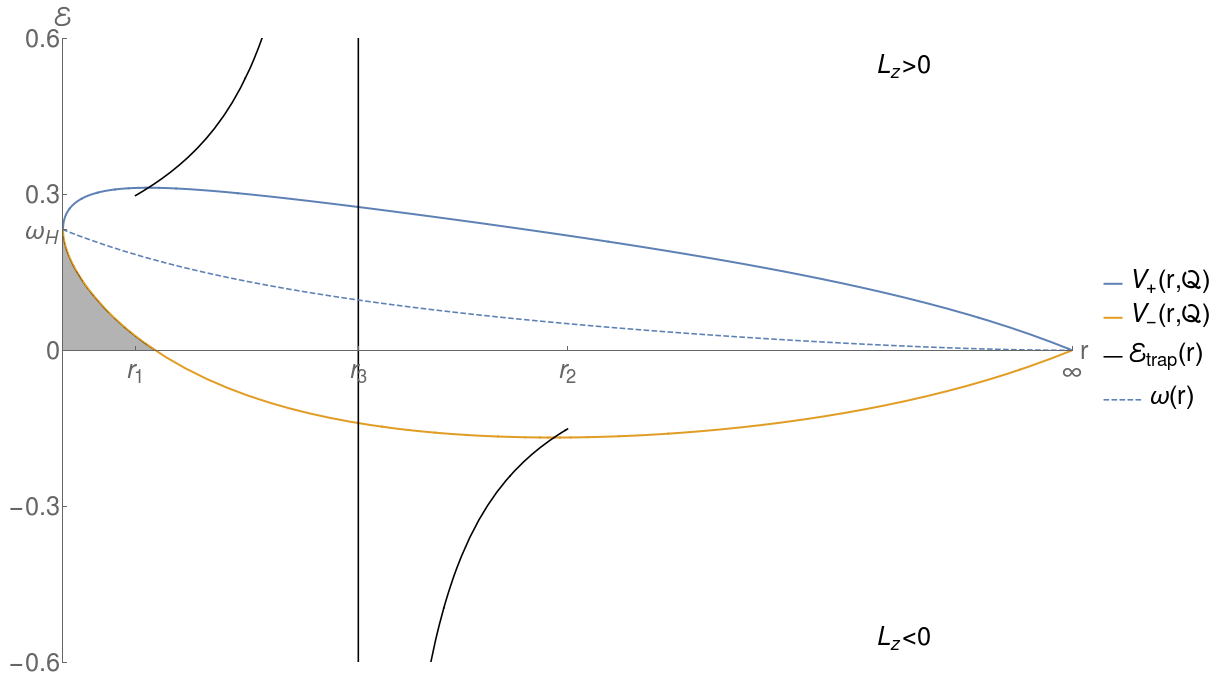}
\caption{Plot of the pseudo potentials $V_{\pm}$ as function of a compactified radial coordinate in the exterior region for $a=0.764$ and $\mathcal Q = 0.18$. Its qualitative features are preserved when $a$ and $\mathcal Q$ are changed. The location of trapping in phase space is indicated by the function $\mathcal E _{trap} (r)$. The extrema of the pseudo potentials are the intersection of $V_+$ and $V_-$ with this function. Therefore they slide on this curve as $\mathcal Q$ increases. The area filled in gray corresponds to geodesics with $E<0$. It is clear from this plot that the regions occupied by geodesics of negative energy and trapped geodesics respectively are disjoint in phase space.}
\label{fig:potentials}
\end{figure}
\noindent Furthermore by the eikonal approximation it is clear, that a massless wave equation should relate to the null geodesic equation in the limit of high frequencies. In \cite{2014arXiv1402.7034D} it is shown that when separating the wave equation $\Sigma\Box \psi=0$  the ODE for the radial function in Schr\"odinger form can be written as:
\begin{equation}
\frac{\mathrm{d}^2 u}{\mathrm{d}r^{* \ 2}} + \left ( \frac{R(r,E=\omega,L_z=m,L^2=\lambda _{lm})}{(r^2+a^2)^2} -F(r) \right) u = 0
\end{equation}
with $F(r)=\frac{\Delta}{r^2 + a^2}(a^2 \Delta +2Mr(r^2-a^2)) \geq 0$ and hence we have the following relations:
\begin{equation}
\omega \sim E, \qquad m \sim L_z,\qquad \lambda _ {lm} \sim L^2.
\end{equation}
When trying to understand the different treatments of different parameter ranges in \cite{2014arXiv1402.7034D} it is helpful to play with the parameters of the pseudo potential in the Mathematica notebook provided with this paper \cite{blazej_marius}. The construction of the different mode currents becomes much more intuitive when thinking about where in Figure \ref{fig:potentials} the corresponding parameters are located. Note that in the high frequency limit the pseudo potentials correspond to the location of $\omega^2-V(r,\omega,m,\Lambda)=0$ and hence the location where the leading contribution to the bulk terms of the $Q^y$ and $Q^h$ currents change their sign. \\
Another interesting observation is that combining the results in section \ref{sec:trapped} and section \ref{sec:tortho}, we can see that to separate trapping from the ergoregion in physical space it is sufficient if we restrict the null geodesics to be either co- or counter-rotating. In the co-rotating case there simply does not exist an ergoregion and the statement is clear. In the counter rotating case trapping is constrained to $r\in(r_3,r_2]$ and $r_3>2M\geq r_{ergo}$ for all Kerr spacetimes with $a<M$. In this direction particularly interesting might be the potential functions $\Psi_\pm$ in \cite{hasse_morse-theoretical_2006} which have interesting properties in physical space.

\subsection*{Acknowledgements} We are grateful to Lars Andersson, Volker Perlick and Siyuan Ma for helpful discussions and their comments on the manuscript.
\appendix
\section{}\label{app:A}
Let $f(x)$ be a smooth function on $[-1,1]$ vanishing at the boundary points with a unique maximum with value $1$ at zero. Hence $f(0)=1$, $f'(0)=0$ and $f''(0)<0$. We then define: 
\begin{align}
g_{1a}(x)&=\arcsin(f(x)): & [-1,0)\rightarrow&[0,\pi/2)\\
g_{2a}(x)&=\pi-\arcsin(f(x)):& [-1,0)\rightarrow&(\pi/2,\pi]\\
g_{1b}(x)&=\arcsin(f(x)):& (0,1]\rightarrow&[0,\pi/2)\\
g_{2b}(x)&=\pi-\arcsin(f(x)):& (0,1]\rightarrow&(\pi/2,\pi]
\end{align}
Note that $g_{1a/b}'(x)=-g_{2a/b}'(x)$. We then calculate: 
\begin{equation}
    \frac{d}{dx}g_{1a}(x)= \frac{f'(x)}{\sqrt{1-f(x)^2}}.
\end{equation}
 Note that both the nominator and denominator vanish as $x$ goes to zero. However to apply the rule of l'Hopital we have to consider the square of the expression. We~then get:
\begin{equation}
    \lim_{x \rightarrow0} \left(\frac{d}{dx}\arcsin(f(x))\right)^2=\lim_{x\rightarrow0} \frac{-f''(x)}{f(x)}=-f''(0). 
\end{equation}
Thus $d/dx (\arcsin(f(x)))|_{x=0}= \sqrt{-f''(0)}$. The sign is chosen based on the fact that $d/dx (\arcsin(f(x)))>0$ for $x\in[1,0)$ . Note that on $[-1,0)$ the derivative of $g_1(x)$ is positive while on $(0,1]$ it is negative. Together this gives us that the function:
\begin{equation}
    g(x)=\begin{cases}
    g_{1a}(x)& \text{ if } x\in [-1,0)\\
    \pi/2&  \text{ if } x=0\\
    g_{2b}(x)&  \text{ if } x\in (0,1]
    \end{cases}
\end{equation}
is smooth at $x=0$ and therefore on $[-1,1]$ with $d/dx(g(x))|_{x=0}=\sqrt{-f''(0)}$. 



\newcommand{\arxivref}[1]{\href{http://www.arxiv.org/abs/#1}{{arXiv.org:#1}}}
\newcommand{\mnras}{Monthly Notices of the Royal Astronomical Society}
\newcommand{\prd}{Phys. Rev. D}
\newcommand{\apj}{Astrophysical J.}

\bibliographystyle{amsplain}
\bibliography{shadow,kerr}

\end{document}